\documentclass[11pt]{article}
\usepackage{efron}

\begin{document}

\title{A bootstrap approach to prediction-powered inference}

\author{Bradley Efron\footnote{Sequoia Hall, 390 Jane Stanford Way, Stanford, CA 94305-4020; \texttt{efron@stanford.edu}}\\\textit{Stanford University}}

\maketitle

\begin{abstract}
Prediction-powered inference (PPI) refers to a two-level situation where the statistician observes a set of $(x,y)$ pairs and another set of $x$s with the responses $y$ missing. Also available is some independent background data from which a prediction rule $f(x)$ has been produced, perhaps by a machine learning algorithm; $f(x)$ approximates $E\{y\mid x\}$ but there is no guarantee of its accuracy for the situation at hand. \cite{angel} developed an algorithm that makes use of all the data, including the unlabeled $x$s, for the estimation of a parameter of interest. A different algorithm is proposed here, using the bootstrap to avoid asymptotics, that is shown to have advantages of efficiency and generality. It is similar in spirit to the original PPI paper by \cite{wang}. Prediction-powered inference raises questions about the information available in unlabeled data, with some surprises here, particularly concerning the estimation of the expected value of $y$.

\bigskip


\end{abstract}

\section{Introduction}\label{sec1}

In an influential paper, \cite{angel} introduced the term ``prediction-powered inference'' or PPI for their treatment of the following situation. We observe a labeled data set of covariate vectors $x$ and $y$,
\begin{equation}
\textit{Labeled data:}\quad\left\{\left(x_a(i),y_a(i)\right),\ i=1,\dots,n_a\right\},
\label{11}
\end{equation}
and also an unlabeled data set where the responses are missing,
\begin{equation}
\textit{Unlabeled data:}\quad\left\{\left(x_b(i)\right),\ i=1,\dots,n_b\right\}.
\label{12}
\end{equation}
(Subscripts $a$ and $b$ will always refer to labeled and unlabeled cases, respectively.)

We wish to estimate some parameter $\theta$, for instance,
\begin{equation}
\theta=E\{y\}\quad\text{or}\quad\theta=\text{correlation}(x,y),
\label{13}
\end{equation}
with the hope of doing better than simply ignoring the unlabeled data and relying entirely on the labeled set. If $n_b$ is much larger than $n_a$ --- a common situation in survey work --- there seems to be a chance for substantial gains hidden in the unlabeled data.

To this end, prediction-powered inference makes a crucial assumption: the existence of a \textit{prediction rule} $f$ based on background data independent of \eqref{11} and \eqref{12}. Ideally, $f(x)$ should accurately approximate the conditional expectation of $y$ given $x$,
\begin{equation}
f(x)=E\{y\mid x\},
\label{14}
\end{equation}
in which case it could very well help infer the missing responses of the unlabeled data. This is a lot to hope for. The ``background data'' may be of uncertain relevance to the current situation. \citeauthor{angel} use a bias-correction device as a remedy (see the examples in \ref{sec3} and \ref{sec5}) while \cite{zrnic} use cross- validation debiasing.

This paper takes a different tack. A straightforward GLM (generalized linear model) approach to prediction-powered inference for binary data (that is, when all the $y$s are 0 or 1) is introduced in \ref{sec2}, with the prediction rule $f$ determining the GLM structure matrix. \ref{sec3} discusses \R1, a bootstrap-based algorithm for carrying out the PPI-GLM calculations when the responses are binary. The important special case $\theta=E\{y\}$ is examined in \ref{sec4}. \ref{sec5} discusses non-binary responses and the corresponding bootstrap algorithm \R2. Side issues of interest are presented as Remarks at the end of each section.

Prediction-powered inference is well suited to current practice, where powerful machine learning algorithms painlessly produce effective prediction rules $f$. As an example, \citeauthor{angel} discuss \textit{Galaxy Zoo 2} \citep{willett}: some 250,000 galaxies were laboriously labeled as either spiral ($y=1$) or non-spiral ($y=0$) by internet volunteers. Using this data, the machine learning algorithm ResNet produced a prediction rule $f$ which then could be used to classify the millions of unlabeled galaxies.

An astronomer investigating a special population of galaxies, perhaps those in a small corner of the sky, might be grateful for the ResNet prediction rule but worry that it could perform inaccurately in the special population. This is where the adjustment procedures in \cite{angel} and this paper come into play. The examples featured in what follows concern real data but smaller data sets, where it is easier to see both the advantages and weaknesses of PPI.

The idea of prediction-powered inference, though not the name, was pioneered in the path-breaking paper by \cite{wang}. Their treatment (which partly relies on bootstrap computations) is close in spirit if not detail to the present paper. The discussion in \ref{sec5} covers similar ground to \citeauthor{wang}, as outlined in \ref{rem12}. \cite{motwani} criticized \citeauthor{wang}'s algorithm on the grounds of estimation bias, asking what parameter it was actually estimating; this question is answered for \R1 and \R2 in \ref{sec3} and \ref{sec5}.

\cite{angel} employ asymptotic approximations to assess the width of their proposed PPI confidence intervals. Improved versions of the asymptotics are considered in \cite{angelx}. \R1 and \R2 use bootstrapping to avoid the asymptotics. \ref{sec3} and \ref{sec5} include comparisons of the two methods.

Going further afield, the bootstrap replicates used in \R1 and \R2 can be thought of as \textit{synthetic data} (\citealp{jordon,shen} and many others). The PPI-GLM framework of \ref{sec2}--\ref{sec4} is useful in demonstrating what's at stake in using synthetic data for statistical inference. A principal goal of what follows is putting prediction-powered inference into the context of classical statistical modeling.

The labeled/unlabeled data structure \eqref{11}--\eqref{12} isn't unique to PPI: an epidemiological study might obtain disease status $y$ for a small portion of its subjects but only covariates $x$ for the great majority. A unique feature of PPI is the employment of prediction algorithms $f$ that now are available from large databases and powerful machine learning algorithms. What follows here is intended to make PPI methods easy to use in a wide variety of applications, while giving readers a clear picture of their strengths and limitations.

\begin{notn}
Boldface symbols will denote vectors and matrices,
\begin{equation}\begin{aligned}
&\text{$\bx_a$ the $n_a\times p$ matrix of labeled covariates,}\\
&\text{$\by_a$ and $\bmf_a$ the $n_a$-vectors of labeled responses and predictors,}
\end{aligned}\label{15}
\end{equation}
and likewise $\bx_b$ and $\bmf_b$.
\end{notn}

\section{A PPI model for binary response data}\label{sec2}

A necessary assumption underlying the PPI model \eqref{11}--\eqref{12} is that the \textit{same} probability distribution gives the labeled and unlabeled data, the only difference being the loss of the labels $y(i)$. We assume that the $(x,y)$ pairs are independent of each other, with joint density say
\begin{equation}
p(x,y)=g(x)f(y\mid x);
\label{21}
\end{equation}
for the unlabeled data, $x$ follows density $g(x)$ while $y$ is missing.

\textit{Binary response} refers to data where each $y$ observation is either 1 or 0, in which case $y$ is \textit{Bernoulli}:
\begin{equation}
y\mid x=\begin{cases}1&\text{with probability }\pi(x)\\0&\text{with probability }1-\pi(x).\end{cases}
\label{22}
\end{equation}
The conditional probability that $y=1$ given $x$ is denoted by $\pi(x)$. ``Classification'' is another name for binary response prediction problems.

For compact notation we write
\begin{equation}
\by_a\sim\bern(\bpi_a)
\label{23}
\end{equation}
to indicate
\begin{equation}
y_a(i)\mid x_a(i)\ind\bern(\pi_a(i))
\label{24}
\end{equation}
for $i=1,\dots,n_a$. In what follows, the unknown vector $\bpi_a$ will be specified as a function of a low-dimensional coefficient vector $\beta$, making \eqref{24} a GLM (though not a logistic regression).

Like all regression models, prediction-powered inference can be carried out either \textit{unconditionally} or \textit{conditionally} on the $x$ values, the latter meaning that the covariates $\bx_a$ and $\bx_b$ are considered fixed as observed. Unconditional inference is more relevant when different prediction models are being compared; conditional inference is appropriate when the current data is of central importance going forward.

Philosophical arguments aside, we will see that conditional procedures can provide substantially shorter confidence intervals. This section focuses on conditional computations because they allow familiar GLM results to be brought to bear. Algorithm \R1 (\ref{sec3}) supports both conditional and unconditional analysis.

We begin with $\bl_a$ the vector of logits of the machine learning estimates $f_a(i)$,
\begin{equation}
\bl_a=\left(\cdots\log\{f_a(i)\}/(1-f_a(i))\cdots\right).
\label{25}
\end{equation}
The model
\begin{equation}
\text{model}_a:\glm(\by_a\sim\bl_a,\text{family}=\text{binomial})
\label{26}
\end{equation}
(in R notation) yields coefficient vector $\hbet$,
\begin{equation}
\hbet=\left(\hbet(0),\hbet(1)\right).
\label{27}
\end{equation}
(Models more complicated than \eqref{26} are considered in \ref{sec3}.) Letting $L_a$ be the $n\times2$ matrix having $i$th row $(1,l_a(i))$, the estimated probability vector $\hbpi_a$ from \eqref{26} is then
\begin{equation}
\hbpi_a=\left(\bone+\exp\left\{-L\hbet\right\}\right)^{-1},
\label{28}
\end{equation}
the inverse logit of $L_a\hbet$. See \ref{rem1} at the end of this section.

There is no guarantee that the background data which generated the machine learning predictor $f$ is fully appropriate for the current situation. The GLM procedure \eqref{26} uses the labeled data to provide $\hbpi_a$, a better estimate of the true vector of probabilities $\bpi_a$.

\begin{figure}[htbp]
\centering
\includegraphics{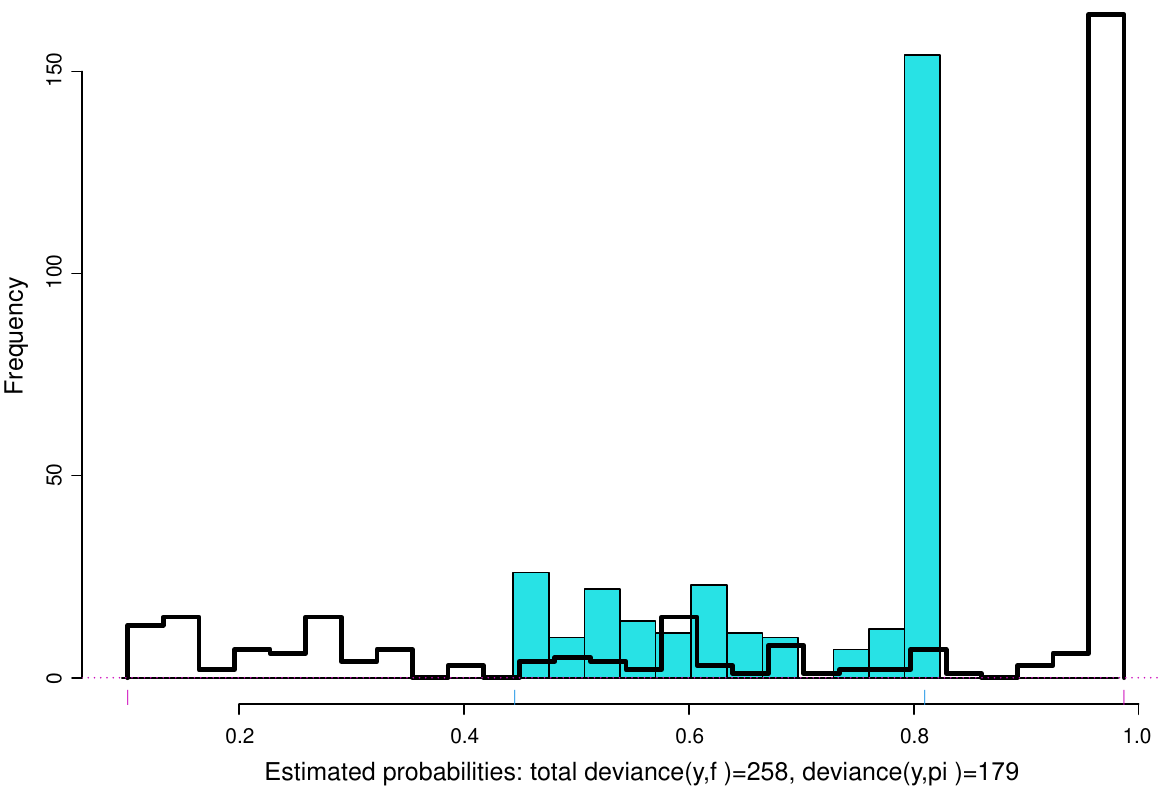}
\caption{Machine learning probabilities $f_a$ (solid histogram) and GLM probabilities \eqref{28} $\pi_a$ (line histogram), $n=300$ cases.}
\label{fig1}
\end{figure}

\ref{fig1} compares the components of $\bmf_a$ with those of $\hbpi_a$ for the labeled dataset ``Pew0'' of an example discussed in \ref{sec3}. Here $n_a=300$. The components of $\hbpi_a$ are more dispersed than those of $\bmf_a$, and give a much better fit to the observed responses $\by_a$: total deviance$(\by_a,\hbpi_a)=179$ compared to deviance$(\by_a,\bmf_a)=258$.

We will consider estimating a class of parameters $\theta$ defined in the following way: beginning with a statistic
\begin{equation}
\hthe=t(\bx,\by),
\label{29}
\end{equation}
the parameter of interest is the expectation of $t(\bx,\by)$ given $\bpi$ and denoted $T(\bpi)$,
\begin{equation}
\theta=T(\bpi)=E_{\bpi}\{t(\bx,\by)\}.
\label{210}
\end{equation}

For instance, if $t(\bx,\by)=\sum y(i)/n$ --- remembering that in this section $x$ is considered fixed --- then
\begin{equation}
\theta=T(\bpi)=\sum(\pi(i))/n,
\label{211}
\end{equation}
which is $E\{y\}$ \eqref{13}. The main example in \ref{sec3} takes
\begin{equation}
t(\bx,\by)=\text{ sample correlation }(\bx,\by),
\label{212}
\end{equation}
which makes $\theta$ multidimensional.

Standard exponential family results, as in Section 3.1 of \cite{2023}, allow us to state PPI in familiar GLM terms. Define $V_a$ to be the diagonal matrix
\begin{equation}
V_a=\diag(v_a(i))\where v_a(i)=\pi_a(i)(1-\pi_a(i)),\ i=1,\dots,n_a.
\label{213}
\end{equation}
The coefficient vector $\hbet$ \eqref{27} has approximate covariance matrix
\begin{align}
\cov\phbet&\doteq G^{-1},\label{214}\\
\intertext{where}
G=L_a'V_aL_a,\label{215}
\end{align}
$L_a=(\bone,\bl_a)$ as before. Letting $\hbh_a=L_a\hbet$ we have $\hbpi_a=\bone/(\bone+\exp\{-\hbh_a\})$, giving
\begin{equation}
d\hbpi_a/d\hbh_a=V_a,
\label{216}
\end{equation}
$\cov(\hbh_a)\doteq L_aG^{-1}L_a'$, and finally
\begin{equation}
\cov\left(\hbh_a\right)\doteq V_aL_aG^{-1}L_a'V_a.
\label{217}
\end{equation}

Now let $\hbdel_a$ be the gradient vector of $\theta=T(\bpi)$ evaluated at $\hbpi_a$,
\begin{equation}
\hbdel_a(i)=\left.\frac{\partial T(\bpi_a)}{\partial\pi_a(i)}\right|_{\hbpi_a}\for i=1,\dots,n_a.
\label{218}
\end{equation}
The delta-method approximate standard deviation of $\hthe_a=T(\hbpi_a)$ is then
\begin{equation}
\sd\left(\hthe_a\right)=\left[\hbdel_a'V_aL_aG^{-1}L_a'V_a\hbdel_a\right]^{1/2}.
\label{219}
\end{equation}

There is a version of \eqref{219} that applies to $\hthe_b$, the estimate of $\theta$ obtained using both the labeled and the unlabeled data. It begins with $\bpi_b$, the vector of probabilities $\pi_b(i)$ for the unlabeled data cases. Letting $\bl_b$ be the vector of logits $l_b(i)=\log\{f_b(i)/(1-f_b(i))\}$, and $L_b$ the matrix $(\bone,\bl_b)$, the estimate $\hbpi_b$ is
\begin{equation}
\hbpi_b=\left(\bone+\exp\left\{-L_\beta\hbet\right\}\right)^{-1},
\label{220}
\end{equation}
with $\hbet$ the \textit{same} as in \eqref{28}.

The estimated standard deviation for $\hthe_b=T(\hbpi_b)$ is the analog of \eqref{219},
\begin{equation}
\sd\left(\hbet_b\right)\doteq\left[\bdel_b'V_bL_bG^{-1}L_b'V_b\bdel_b\right]^{-/1},
\label{221}
\end{equation}
with $G$ as given in \eqref{215}; in \eqref{221},
\begin{equation}
\bV_b=\diag(V_b(i)),\qquad v_b(i)=\hpi_b(i)(1-\hpi_b(i))\text{ for } i-1,\dots,n_b,
\label{222}
\end{equation}
and $\hbdel_b$ the gradient vector of $\hthe_b$,
\begin{equation}
\hdel_b(i)=\frac{\partial T(\hbpi_b)}{\partial\hpi_b(i)}\for i=1,\dots,n_b.
\label{223}
\end{equation}

The approximate confidence interval for $\theta_b=T(\bpi_b)$,
\begin{equation}
\theta_b\in\hthe_b\pm c\sd\left(\hthe_b\right)
\label{224}
\end{equation}
($c=1.96$ for approximate 95\% coverage), has the full weight of maximum likelihood theory behind it. That is, if we start from model$_a$ and ignore the many approximations going into \eqref{224} it is, to paraphrase Einstein, as short as possible but not any shorter. \R1 in \ref{sec3} uses bootstrap methods to avoid the approximations going into \eqref{221}.

\begin{figure}[htbp]
\centering
\includegraphics[keepaspectratio, width=0.6\linewidth, clip]{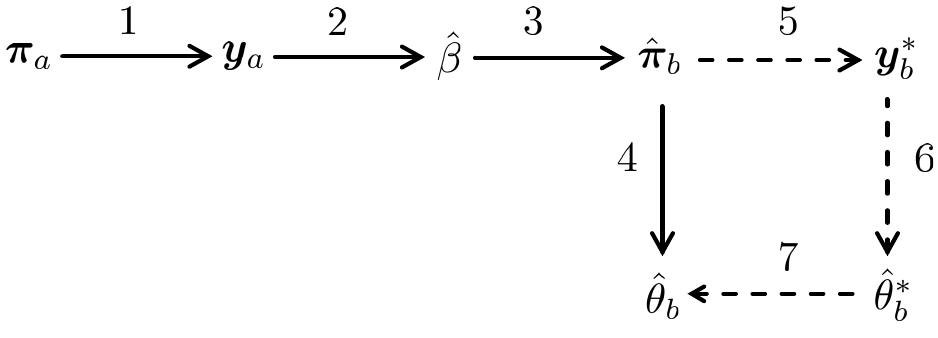}
\caption{Steps in the calculation of $\hbet_b$, as explained in the text, \eqref{220}--\eqref{221}.}
\label{fig2}
\end{figure}

\ref{fig2} diagrams the calculation of $\hthe_b$ in four steps:
\begin{enumerate}
\item Observe $\by_a$, a Bernoulli random sample \eqref{23} from the unknown $\bpi_a$.
\item Calculate $\hbet$ from \eqref{26}--\eqref{27}.
\item Evaluate $\hbpi_b$ \eqref{220}.
\item Calculate $\hthe_b=T(\hbpi_b)$.
\end{enumerate}

\noindent That is, $\hthe_b$ is the MLE of $\theta_b=T(\hbpi_b)$ having observed $\by_a$ and \eqref{221} is its usual delta-method estimated standard deviation.

Unfortunately, the four-step program has a fatal flaw, at least as far as routine applications are concerned: the function $\hthe=T(\hbpi)$ is \textit{not} usually available nor is its gradient $\hbdel$ \eqref{223}, blocking step 4 in \ref{fig1}. (See \ref{sec4}.)

The dashed lines show an alternate path:
\begin{enumerate}[resume]
\item Some large number $B$ of parametric bootstrap replications $\bya_b\sim\bern(\hbpi_b)$ are generated.
\item Each gives a bootstrap replication $\hthea_b=t(\bx_b,\bya_b)$ \eqref{29}.
\item $\hthe_b=\sum_1^B\hthea_b/B$.
\end{enumerate}

\noindent As a bonus, the empirical standard deviation $\hsd_b$ of the $\hthea_b$ values estimates the variability of $\hthe_b$.

The crucial feature of model$_a$ is the use of $\bl_a=\logit(\bmf_a)$ as the structure vector in GLM \eqref{26}. (More elaborate models are available; see \ref{rem7} in \ref{sec3}.) This permits the labeled data to adapt the background prediction rule $f$ to the current situation, as in \ref{fig1}. The Pew Research example of \ref{sec3} supports the use of model$_a$.

\begin{rem}
The GLM fitting process \eqref{26} has a forgiving aspect: suppose that $\bl_a$ in \eqref{25} has been distorted by a hidden linear transformation,
\begin{equation}
\tbl_a=c+d\bl_a.
\label{225}
\end{equation}
Nevertheless, the resulting estimate $\hbpi_a$ \eqref{28} stays the same.
\label{rem1}
\end{rem}

\section{Bootstrap implementation of prediction-powered classification}\label{sec3}

The diagram in \ref{fig2} describes the calculation of $\hthe_b$, our estimate of $\theta$ based on both the labeled and unlabeled data \eqref{11}--\eqref{12}. \R1 is a bootstrap-based program for calculating $\hthe_b$ and its estimated standard deviation $\hsd_b$ when the responses $y(i)$ are binary \eqref{22}. It automates the construction of approximate confidence intervals \eqref{224} from PPI data.

There are three notable advantages to using bootstrap methodology for PPI estimation:
\begin{enumerate}
\item \R1 applies to any statistic $t(\bx,\by)$ \eqref{29}.
\item It avoids asymptotic approximations such as those going into \eqref{219} or \citeauthor{angel}'s formulas.
\item The resulting confidence intervals \eqref{224} are based on familiar GLM theory and as such have, theoretically at least, close to minimal length.
\end{enumerate}

\noindent This section analyzes the logic behind bootstrap PPI confidence intervals. We begin first with a dataset and example featured in what follows.

A Pew Research Center poll from 2020 reported on the approval ($y=1$) or disapproval ($y=0$) of President Biden's COVID-19 messaging. Each respondent answered 10 questions relating to aspects of the pandemic, resulting in a 10-vector $x$ of predictors. A training set of answers from several thousand respondents used the boosting algorithm \texttt{XGboost} to produce a prediction rule $f(x)$; it had an error rate of 15\% on the training data.

A separate collection of 6000 $(x,y,f)$ cases was used to create versions of the labeled and unlabeled datasets \eqref{11}--\eqref{12}: \textit{Pew0} is the main example used in what follows and consists of $n_a=300$ $(x,y)$ pairs and 900 $x$s randomly selected from the 6000 cases; \textit{Pew1:100} comprises 100 additional randomly selected labeled/unlabeled datasets, each with $n_a=300$ and $n_b=900$. Each of the 101 examples has matrix $\bx_a$ 300 by 10, dimension 300-vectors $\by_a$ and $\bmf_a$, matrix $\bx_b$ 900 by 10, and 900-vector $\bmf_b$ (with $\by_b$ taken to be missing).

\begin{table}[htbp]
\caption{Point estimates of $\theta=\corr(x,y)$ for 10 questions in Pew0 data: line 1 from \R1 uses both labeled and unlabeled data as well as machine learning predictions; line 2 from \R1 does not use the unlabeled data; line 3 is the classical estimator $\corr(\bx,\by)$ and doesn't use unlabeled data or machine learning predictions. There are no significant differences among the three sets.}
\begin{center}\begin{small}\begin{tabular}{lcccccrcccc}
                        &     Q1&     Q2&     Q3&     Q4&     Q5&    Q6&     Q7&     Q8&     Q9&   Q10\\\hline
1.\ Lab \& unlab&  $.65$& $-.24$& $-.50$& $-.27$& $-.10$&$-.08$& $-.13$& $-.23$& $-.16$& $-.10$\\
2.\ Lab only&      $.66$& $-.26$& $-.49$& $-.25$& $-.06$&$ .00$& $-.23$& $-.23$& $-.14$& $-.07$\\
3.\ Classical&     $.62$& $-.31$& $-.54$& $-.26$& $-.14$&$ .04$& $-.23$& $-.30$& $-.19$& $-.06$
\end{tabular}\end{small}\end{center}
\label{tab1}
\end{table}

As a first example, described more fully below, \R1 was applied to Pew0 with
\begin{equation}
t(\bx,\by)=\text{correlation}(\bx,\by)
\label{31}
\end{equation}
as the statistic of interest. \ref{tab1} shows the reported estimates of $\theta=E_\pi\{t(\bx,\by)\}$ \eqref{210}, calculated using different portions of Pew0:
\begin{enumerate}
\item The top line uses all of Pew0, both labeled and unlabeled.
\item The middle line uses only the labeled data $\bx_a$, $\by_a$, and $\bmf_a$, ignoring $\bx_b$ and $\bmf_b$.
\item The bottom line --- the ``classical'' estimate --- is the sample correlation $t(\bx,\by)$, now ignoring $\bmf_a$ as well.
\end{enumerate}

\noindent The three estimates are nearly the same: none of the entries for any question come close to a significant difference.

\begin{figure}[htbp]
\centering
\includegraphics{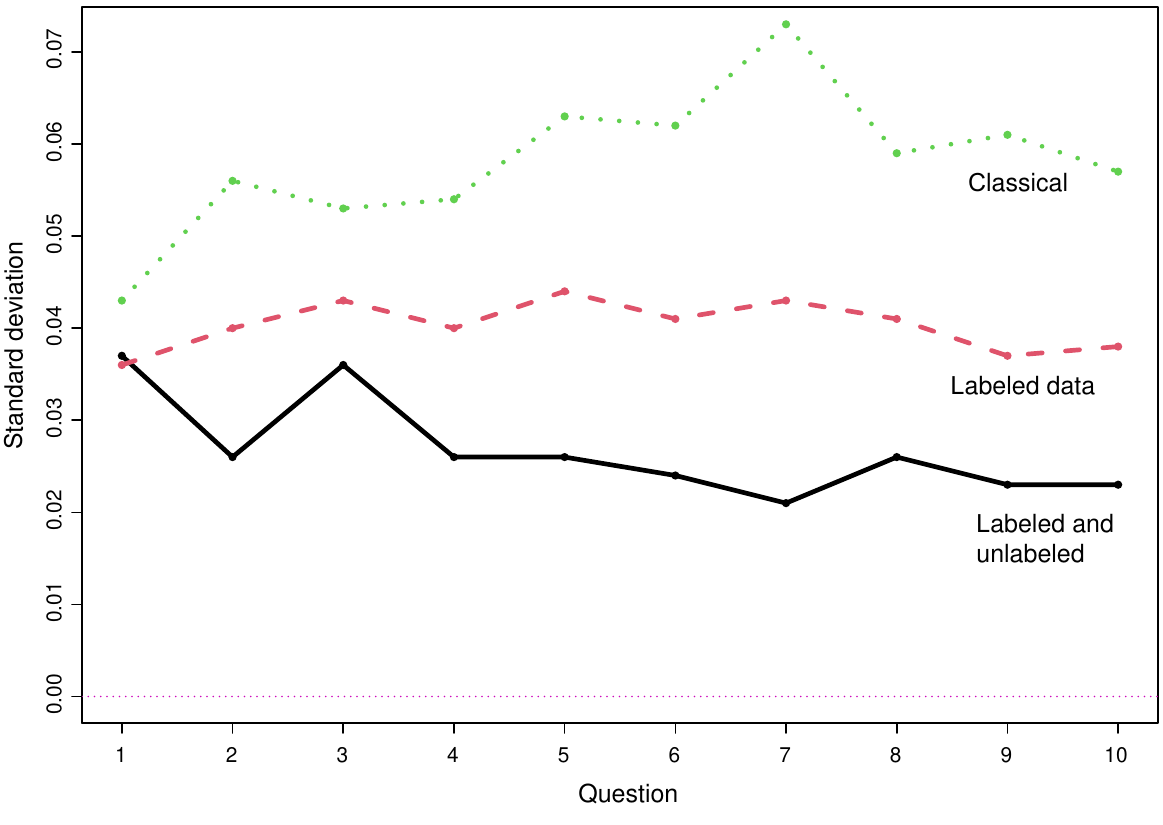}
\caption{\R{1} estimated standard deviations for $\corr(x,y)$, Pew0; classical (top), labeled data (middle), labeled and unlabeled (bottom).}
\label{fig3}
\end{figure}

The width of a confidence interval $\hthe\pm\hsd$ depends on its standard deviation $\hsd$. \ref{fig3} compares the estimated standard deviations for the three techniques in \ref{tab1}, labeled and unlabeled, labeled only, and classical. Summarizing results for the 10 questions, labeled and unlabeled standard deviations are about half as large as the classical values and about two-thirds as large as those using labeled only.

How good is the performance seen in \ref{fig3}? If we actually had the labels for the 900 unlabeled cases it would make 1200 $(x,y)$ pairs in all, so $t(\bx,\by)$ using all the data would give about half the standard deviation of the classical values shown in \ref{fig1} --- roughly matching the result seen in \ref{fig3} --- which looks like an excellent performance by \R1. More realistically, given the way \R1 will be shown to operate, we might hope for a reduction factor of $\sqrt{300/900}=0.58$ between labeled only and labeled plus unlabeled, which is not strikingly worse than the observed two-thirds. Somewhat sobering, however, is the behavior on question 1 where the labeled-and-unlabeled standard deviation actually exceeds the labeled-only sd. \ref{sec4} has more to say about the efficiency of PPI methodology.

The purpose of \R1 is to provide the point estimate $\hthe_b$ diagrammed in \ref{fig2}, as well as its standard error. This task can be carried out either conditionally (on the $x$ values) or unconditionally. The explanation next assumes conditional inference, with unconditional analysis discussed in \ref{rem4} at the end of this section.

\begin{figure}[htbp]
\centering
\includegraphics[keepaspectratio, width=0.8\linewidth, clip]{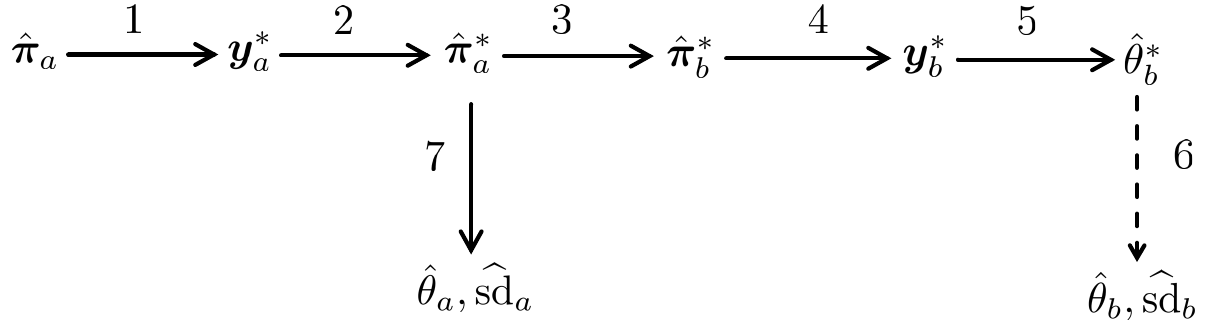}
\caption{Top row shows steps in the calculation of a single bootstrap replication $\hthea_b$, as explained in the text. Multiple replications then give estimate $\hthe_b$ and standard deviation $\hsd_b$ (step 6), as well as corresponding values $\hthe_a$ and $\hsd_a$, which ignore the unlabeled data (step 7).}
\label{fig4}
\end{figure}

The top line of \ref{fig4} diagrams the computation of a single bootstrap replication $\hthea_b$ beginning with $\hbpi_a$ \eqref{28}, the estimate of $\pi$ obtained from model$_a$ \eqref{26} (asterisks indicate bootstrap quantities):
\begin{enumerate}
\item Sample $\bya_a\sim\bern(\hbpi_a)$.
\item Calculate $\hbpia_a$ as in \eqref{26}--\eqref{28} with $\bya_a$ replacing $\by_a$.
\item Select $\hbpia_b$ such that the graph of $(\bmf_b,\hbpia_b)$ matches that of $(\bmf_a,\hbpia_a)$. (See \ref{rem5}.)
\item $\bya_b\sim\bern(\hbpia_b)$.
\item $\hthea_b=t(\bx_b,\bya_b)$.
\end{enumerate}

\noindent The quantity $\hthea_b$ is a single bootstrap replication of the estimate $\hthe_b$ in \ref{fig2}. Creating some large number $B$ of such replicates ($B=1000$ by default) provides an estimate and standard deviation for the parameter of interest $\theta$, step 6,
\begin{subequations}\begin{align}
\hthe_b&=\sum_{j=1}^B\frac{\hthea_b(j)}{B}\label{33a}\\
\intertext{and}
\hsd_b&=\left[\sum_{j=1}^B\frac{\left(\hthea_b(j)-\hthe_b\right)^2}{B-1}\right]^{1/2};\label{33b}
\end{align}\label{33}
\end{subequations}
these are the numbers reported in \ref{tab1} and \ref{fig3}.

In any one application, two factors determine the performance of prediction-powered inference: the efficacy of $\bmf$ as an estimate of $\bpi$, and the information for estimating $\theta$ available in the unlabeled cases. Step 7 in \ref{fig4} helps separate the effects of these two factors: changing subscript $b$ to $a$ in steps 4--6 gives an estimate $\hthea_a$ that depends only on the labeled data $\bx_a$, $\by_a$, and $\bmf_a$. The ``labeled data'' line in \ref{fig3} shows that supplementing $\bx_a$ and $\by_a$ with $\bmf_a$ provides about half the improvement over ``classical'', with the other half coming from using $\bx_b$ and $\bmf_b$. \ref{sec4} investigates the two factors for the case $\theta=E\{y\}$.

\begin{figure}[htbp]
\centering
\includegraphics{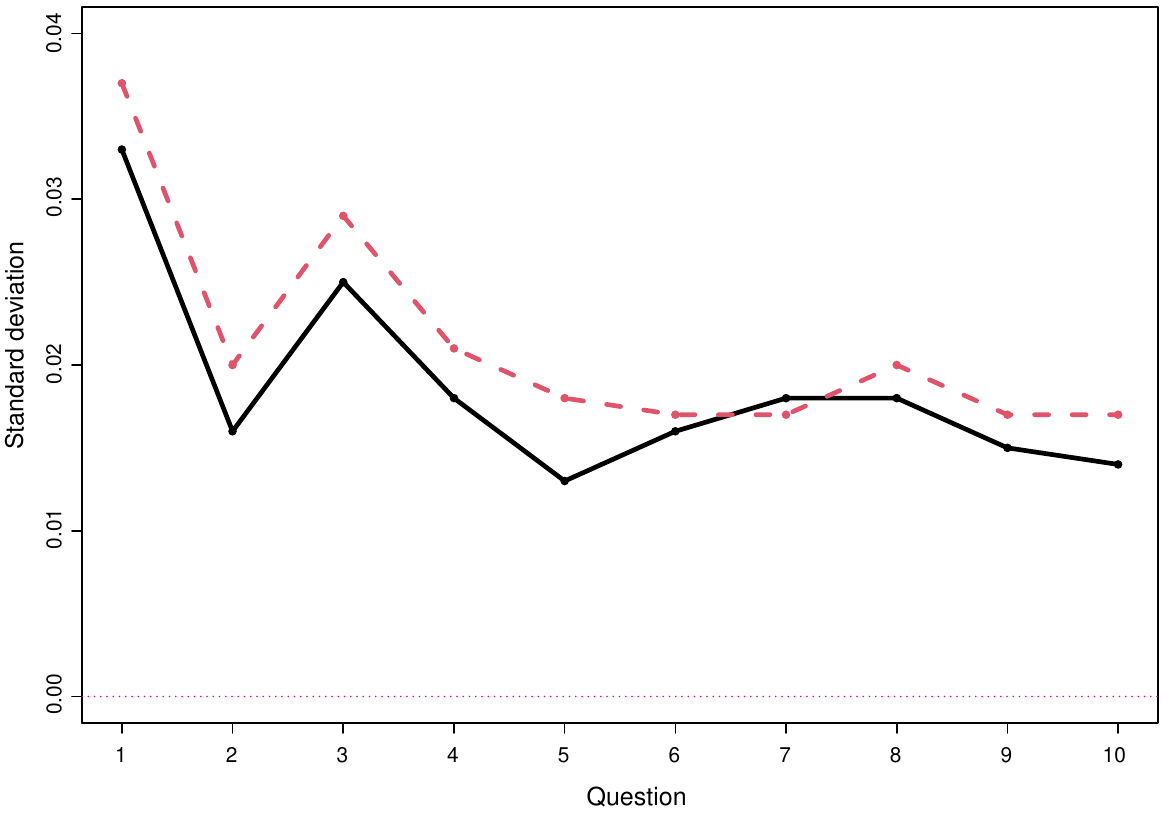}
\caption{Standard deviation estimates for $\corr(x,y)$, 100 trials Pew1:100; empirical sd (solid) and average of $\hsd_b$ from \R1 (dashed).}
\label{fig5}
\end{figure}

As a check on the performance of the estimated sd for $\theta=\corr(\bx,\by)$, \R1 was run for the 100 randomly selected datasets Pew1:100 described at the start of this section. Each dataset gave a $\hthe_b$ and $\hsd_b$. \ref{fig5} compares the ``external variability'' $S_{\extr}=$ empirical standard deviation of the 100 $\hthe_b$ values with the ``internal variability'' $S_{\intr}=$ mean of the 100 $\hsd_b$ values. $S_{\intr}$ tracks $S_{\extr}$ reasonably well, averaging about 15\% greater, which suggests that confidence intervals $\hthe_b\pm c\times\hsd_b$ may be mildly conservative. \ref{rem6} says more about this.

The PPI algorithm in \cite{angel} operates quite differently from \R1. A single example follows: Suppose that the statistic of interest $\theta=t(\bx,\by)$ is the vector of logistic regression coefficient of $y$ on $x$, calculated as
\begin{equation}
t(\bx,\by)=\glm(\by\sim\bx,\text{binomial})\$\text{coef}
\label{34}
\end{equation}
in R. We will ignore the intercept coefficient, making $\theta$ 10-dimensional for the Pew Research Center data. Notice we are not assuming that the logistic regression model in \eqref{34} is correct.

Following the formulation in \cite{angelx}, for a trial value $\hthe$ of $\theta$ we define the loss function
\begin{equation}
L\left(\tthe\right)=D\left(\by,\bx_a\cdot\tthe\right)+\frac12\left[D\left(\bmf_b,\bx_b\cdot\tthe\right)-D\left(\bmf_a,\bx_a\cdot\tthe\right)\right],
\label{35}
\end{equation}
where $D(\by,\bz)$ is the average Bernoulli deviance between $\by$ and $\bz$. (See \ref{rem6}.) The PPI estimate of \cite{angel} is
\begin{equation}
\htheang=\argmin\left\{L\left(\tthe\right)\right\},
\label{36}
\end{equation}
the value of $\tthe$ minimizing \eqref{35}. ($\htheang$ corresponds to $\hthe_b$ \eqref{33}.) The authors show that $\hthe$ is asymptotically unbiased, with an estimated standard deviation based on large-sample calculations.

\begin{table}[htbp]
\caption{Average point estimates of logistic regression coefficients for the 100 datasets Pew1:100. \cite{angel} and \R1 estimates are nearly the same.}
\begin{center}\begin{small}\begin{tabular}{lccrcrccccc}
            &    Q1&     Q2&      Q3&     Q4&    Q5&   Q6&     Q7&     Q8&     Q9&    Q10\\\hline
Angelopoulos&$3.73$& $-.23$& $- .87$& $-.07$&$-.17$&$.16$& $-.41$& $-.24$& $-.41$& $-.35$\\
\R1         &$3.43$& $-.12$& $-1.02$& $-.01$&$ .14$&$.01$& $-.19$& $-.22$& $-.37$& $-.28$
\end{tabular}\end{small}\end{center}
\label{tab2}
\end{table}

The estimates $\htheang$ \eqref{36} and $\hthe_b$ \eqref{33} for the logistic regression coefficients \eqref{34} were calculated for the 100 datasets Pew1:100. \ref{tab2} shows that the averages of the two algorithms were nearly the same.

\begin{figure}[htbp]
\centering
\includegraphics{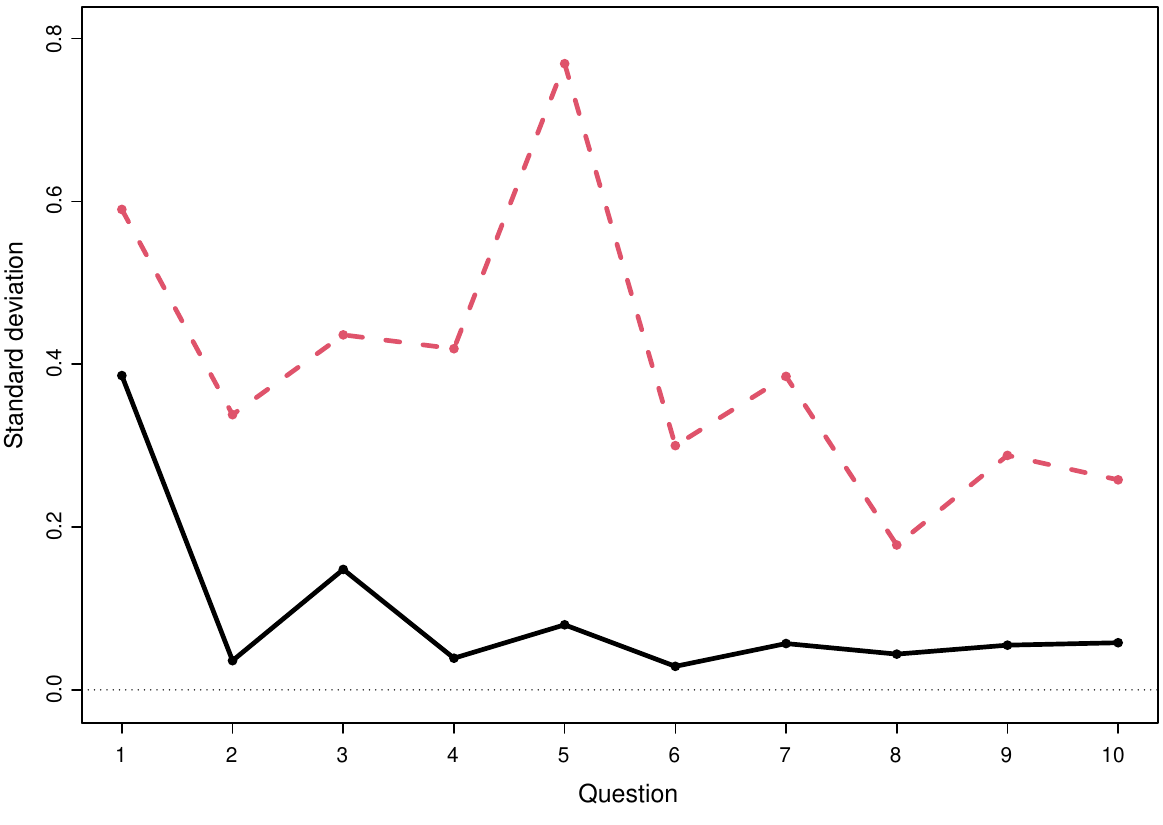}
\caption{Empirical standard deviations over the 100 random datasets Pew1:100; \R1 (solid) and \cite{angel} (dashed).}
\label{fig6}
\end{figure}

\ref{fig6} graphs the empirical standard deviations of the two pairs of 100 estimates for each of the 10 questions. \R1 was much more accurate in this case, with sd values about 20\% as large overall. (The \citeauthor{angel} algorithm performs better in the non-binary example of \ref{sec5}.)

\cite{angel} use the machine-learning predictions $\bmf_a$ and $\bmf_b$ quite differently from \R1: $\bmf_b$ substitutes for the missing response vector $\by_b$ in \eqref{35}, compared with \eqref{25}--\eqref{26} where $\bmf_a$ plays the role of a structural vector (like $x$ in a simple $(x,y)$ regression). Their approach has the advantage of being nonparametric and assumption-free on the role of $f$, but may be inefficient, as in \ref{fig6}.

\begin{rem}
Looking back at \ref{fig4} in \ref{sec3}, suppose the function $\theta=T(\bpi)$ was easily available. This would greatly simplify the bootstrap algorithm in \ref{fig4}: after step 3 we could directly calculate $\hthea_b=T(\hbpia_b)$; $B$ such replications would yield $\hthe_b$ and $\hsd_b$ as in \eqref{33}, which is to say that we could carry out a standard parametric bootstrap analysis of PPI estimation.
\label{rem2}
\end{rem}

\begin{rem}
Usually the function $T\pdot$ isn't available, necessitating steps 4 through 6 in \ref{fig4}. There are \textit{two} simulations in the top line of \ref{fig4}, after steps 1 and 4, mimicking steps 1 and 5 in \ref{fig2}. It might seem better to simulate \textit{many} replications $\hthea_b$ in steps 4 and 5 of \ref{fig4}, instead of just one, taking their average as a better version of $\hthe_b$ for use in \eqref{33}. However, besides raising the computational burden, this would invalidate the standard deviation calculation in \eqref{33}: we originally get to observe only a \textit{single} vector $\by_a$ at the beginning of \ref{fig2}, not many of them; the averaged versions of $\hthea_b$ suggested above would be misleadingly accurate as far as inference is concerned.
\label{rem3}
\end{rem}

\begin{figure}[htbp]
\centering
\includegraphics{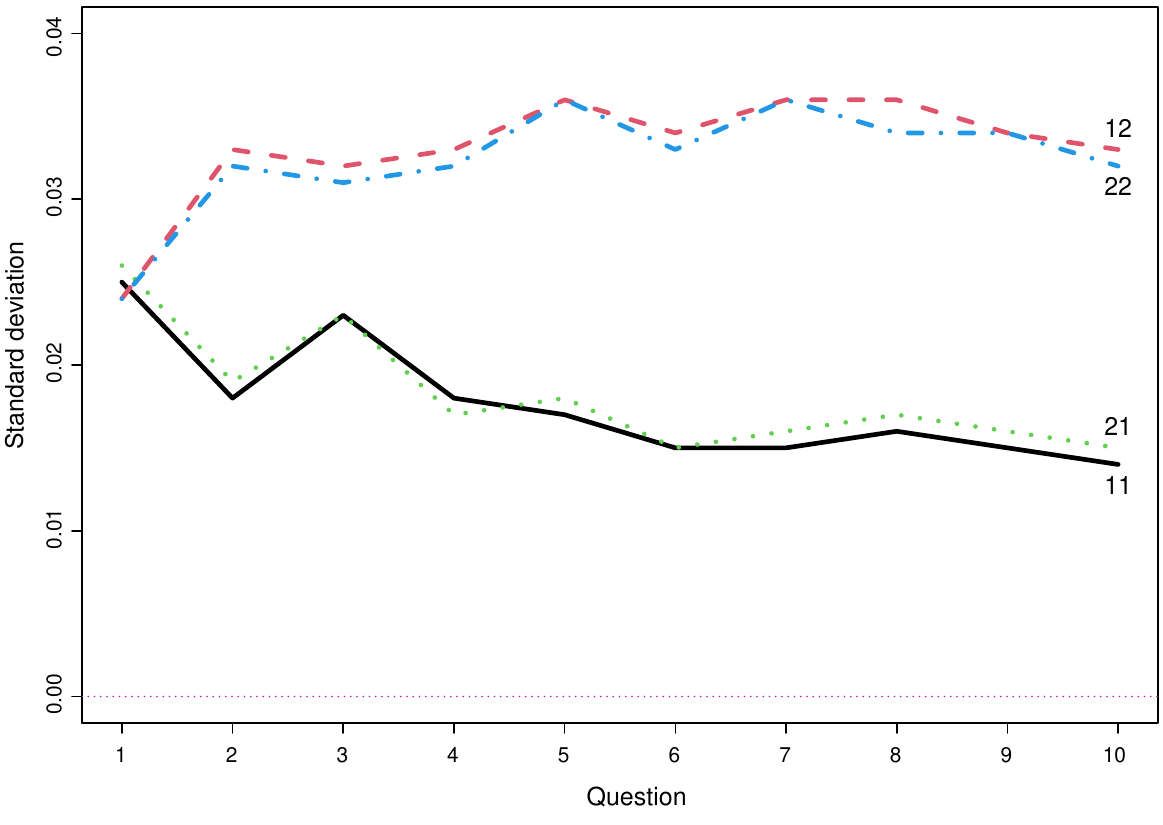}
\caption{Standard deviation estimates for $\corr(x,y)$, Pew0 data; ``22'' unconditional at both levels, ``12'' conditional then unconditional, etc.}
\label{fig7}
\end{figure}

\begin{rem}
The calculations for the preceding tables and figures were carried out conditionally, that is, with $\bx_a$ and $\bx_b$ considered fixed as observed. \R1 also allows unconditional analysis: letting
\begin{equation}
z_a(i)=(x_a(i),y_a(i),f_a(i))\for i=1,\dots,n,
\label{37}
\end{equation}
step 1 in \ref{fig4} is replaced by $\hbpi_a\rightarrow\bza_a$, where $\bza_a$ indicates a random sample with replacement from $\{z_a(1),z_a(2),\dots,z_a(n)\}$. This makes the first step in \ref{fig4} unconditional rather than conditional. A similar construction at step 3 makes the second bootstrap calculation unconditional.

\R1 was applied to the Pew0 data with $t(\bx,\by)=\corr(\bx,\by)$. All four possibilities of conditional or unconditional were tried in \ref{fig7}, coded ``12'' for ``step1 conditional and step 3 unconditional'' and likewise ``11'', ``21'', and ``22''. The means of the point estimates were indistinguishable from each other. However, \ref{fig7} shows much larger standard deviations if step 3 is carried out unconditionally; the choice of unconditional or conditional at step 1 made little difference.
\label{rem4}
\end{rem}

\begin{rem}
Formulas \eqref{28} and \eqref{220} in \ref{sec2} guarantee that the graph of $\hpi_b$ versus $f_b$ exactly matches the graph of $\hpi_a$ versus $f_a$. This reflects the assumption that for case $i$, $\Pr\{y(i)=1\mid x(i)\}$ is a function only of its predicted value $f(i)$, irrespective of being in group \eqref{11} or \eqref{12}. At step 3 in \ref{fig4}, \R1 carries out the matching numerically: in R notation,
\begin{equation}
\hbpia_b=\text{approx}(\bmf_a,\hbpia_a,\bmf_b)
\label{38}
\end{equation}
\label{rem5}
\end{rem}

\begin{rem}
The Bernoulli deviance between a single binary response $y$ and parameter value $\pi$ is
\begin{equation}
-2\left[y\log\{\pi\}+(1-y)\log\{1-\pi\}\right].
\label{39}
\end{equation}
Maximizing the average of \eqref{39} over the labeled data is equivalent to maximizing log likelihood, as in equation (5) of \cite{angel}. The authors consider using values other than 0.5 in \eqref{35} as a further optimization step.
\label{rem6}
\end{rem}

\begin{table}[htbp]
\caption{Total binomial deviance for increasing choices of $p_0$ \eqref{310} and $\aic=\dev+2p_0$. (``0'' is $\dev(\by_a,\bmf_a)$.). No choice significantly improves upon $p_0=1$.}
\begin{center}\begin{small}\begin{tabular}{lccccc}
$p_0$&  0&    1&    2&    3&    4\\\hline
Dev&  258&  179&  177&  172&  170\\
AIC&  258&  181&  181&  178&  178
\end{tabular}\end{small}\end{center}
\label{tab3}
\end{table}

\begin{rem}
Model \eqref{26} assumes that the logit of $\pi_a$ is a linear logistic transform of $l_a=\logit\{f_a\}$. \R1 allows broader models where $\logit\{\pi_a\}$ is some linear function of powers of $l_a$,
\begin{equation}
\glm(\by_a\sim\text{poly}(\bl_a,p_0))
\label{310}
\end{equation}
in R notation, with $p_0$ the top power. In all the preceding examples, $p_0=1$, i.e., \eqref{26}. \ref{tab3} shows the total binomial deviance between $\by_a$ and $\hbpi_a$ for increasing choices of $p_0$, along with AIC $=$ deviance $+2p_0$. No choice is a significant improvement on $p_0=1$.
\label{rem7}
\end{rem}

\begin{table}[htbp]
\caption{External/internal ratios $r$ \eqref{313} and endpoints of approximate 95\% confidence intervals. Data Pew1:100. The estimates for level $b$ are mostly below 1.00, while those for level $a$ mostly vary around 1.00.}
\begin{center}\begin{small}\begin{tabular}{lrrrrrrrrrr}
Level $b$:\\
     &  Q1&   Q2&  Q3& Q4&  Q5&   Q6&   Q7&   Q8&   Q9& Q10\\\hline
$r$&   .88&  .91& .76& .8& .84&  .93&  .98& 1.03&  .93& .82\\
lower& .75&  .78& .66& .7& .74&  .79&  .86&  .88&  .82& .71\\
upper&1.01& 1.04& .84& .9& .93& 1.07& 1.10& 1.19& 1.04& .93\\
\\     
Level $a$:\\
     &   Q1&   Q2&   Q3&   Q4&   Q5&   Q6&   Q7&   Q8&   Q9&  Q10\\\hline
$r$&   1.01&  .99&  .97& 1.03&  .90& 1.01& 1.27& 1.15& 1.06&  .96\\
lower&  .86&  .87&  .85&  .88&  .79&  .88& 1.09& 1.00&  .92&  .82\\
upper& 1.18& 1.11& 1.10& 1.17& 1.01& 1.14& 1.46& 1.29& 1.19& 1.10
\end{tabular}\end{small}\end{center}
\label{tab4}
\end{table}

\begin{rem}
Suppose $\hthe_j$ are independent estimates with means and variances
\begin{equation}
\hthe_j\sim(\mu_j,\sigma_j^2)\for j=1,\dots,N.
\label{311}
\end{equation}
Define the external standard deviation to be
\begin{equation}
S_{\extr}=\left[\frac{E\sum_{j=1}^N\left(\hthe_j-\hthe_\cdot\right)^2}{N-1}\right]^{1/2},\qquad\hthe_\cdot=\frac{\sum_{j=1}^N\hthe_j}{N},
\label{312}
\end{equation}
and the internal standard deviation to be
\begin{equation}
S_{\intr}=\left[\frac{\sum_1^N\sigma_j^2}{N}\right]^{1/2}.
\label{313}
\end{equation}
Then it is straightforward to derive their ratio $r$ as
\begin{equation}
r=\frac{S_{\extr}}{S_{\intr}}=1+\left[\frac{\sum_{j=1}^N(\mu_j-\barmu)^2/(N-1)}{S_{\intr}^2}\right]^{1/2},
\label{314}
\end{equation}
with $\barmu=\sum_{j=1}^N\mu_j/N$.

The top half of \ref{tab4} shows $r$ for $\hthe_b$ of \ref{fig4}, applied to the $N=100$ datasets of Pew1:100. The estimates for the 10 questions tend to have $r$ \text{less than} 1.00, as graphed in \ref{fig5}. This looks impossible according to \eqref{314}, but can happen if the bootstrap estimates $\hsig_{bj}^2$ are larger than the true standard deviations $\sigma_j^2$.

As a check, the bottom half of \ref{tab4} shows $r$ for $\hthe_a$, the level $a$ estimates in \ref{fig4}. Now $r$ varies, mostly, around $r=1.00$. All of this suggests:
\begin{itemize}
\item The bootstrap estimates $\hsig_{aj}$ are close to the true standard deviations $\sigma_a$.
\item The estimates $\hsig_{bj}$ are larger, by perhaps some 15\%.
\item The true means $\mu_j$ are nearly the same for the 100 datasets.
\end{itemize}
\label{rem8}
\end{rem}

\section{The estimation of $\theta=E\{y\}$}\label{sec4}

PPI literature has taken a special interest in the estimation of the parameter $\theta=E\{y\}$, which for binary response is the probability of obtaining $y=1$. Part of the reason is mathematical tractability, another scientific priority. In the Galaxy Zoo 2 example mentioned in \ref{sec1}, $\theta$ is the proportion of galaxies that are spirals, a parameter of theoretical interest. Applying \R1 to $\theta=E\{y\}$ yields somewhat surprising results, discussed next. As before, that discussion will be in conditional terms with $\bx_a$ and $\bx_b$ considered fixed.

The parameter $T(\bpi)=E_{\bpi}\{\by\}$ is unusually amenable in having an explicit form for the function $T$,
\begin{equation}
T(\bpi)=\frac1{n}\sum_1^n\pi(i).
\label{41}
\end{equation}
This allows us to short-circuit the bootstrap diagram in \ref{fig4}: after step 3 we can calculate $\hthe_b$ directly,
\begin{equation}
\hthea_b=\frac1{n_b}\sum_1^{n_b}\tpi_b^*(i).
\label{42}
\end{equation}
This greatly simplifies the bootstrap analysis: drawing a large number of bootstrap replications \eqref{42} gives $\hthe_b$ and $\hsd_b$, as in \eqref{33}, or likewise $\hthe_a$ and $\hsd_a$ if we use only the labeled data $\bx_a$, $\by_a$, $\bmf_a$.

\begin{table}[htbp]
\caption{Estimates and standard deviations for $\theta=E\{y\}$, Pew0 data; ``a'' uses only labeled data including $\bmf_a$, ``b'' also uses unlabeled data. ``Direct'' uses \eqref{41}--\eqref{42}; \R1 uses the full algorithm in \ref{fig4}. ``Classical'' is the mean and standard deviation estimate \eqref{43b}.}
\begin{center}\begin{small}\begin{tabular}{lccccc}
       &\multicolumn{2}{c}{Direct}&\multicolumn{2}{c}{\R1}&Classical\\
       & $a$&  $b$&  $a$&  $b$&     \\\hline
$\hthe$&.730& .735& .730& .735& .730\\
Sd     &.018& .018& .017& .020& .026
\end{tabular}\end{small}\end{center}
\label{tab5}
\end{table}

\ref{tab5} reports on five estimates of $\theta=E\{y\}$ for the Pew0 data of \ref{sec3}, $n_a=300$ and $n_b=900$.  The Classical estimate and standard deviation estimate, based only on $\by_a$, is
\begin{subequations}\begin{align}
\hthe&=\frac1{n}\sum y_a(i)\label{43a}\\
\intertext{and}
\hsd&=\left[\sum_1^{n_a}\frac{\left(y_a(i)-\hthe\right)^2}{n_a-1}\right]^{1/2}\label{43b}
\end{align}\label{43x}
\end{subequations}
``Direct'' is based on $B=1000$ bootstrap replications of $\hthea_a=T(\hbpia_a)$ and of $\hthea_b=T(\hbpia_b)$ as obtained at steps 2 and 3 of \ref{fig4}. ``\R1'' is from \R1 with $B=1000$.

The five point estimates are nearly identical, all close to the mean 0.730 of $y$, which implies low bias for the bootstrap procedures. As far as accuracy is concerned, the Classical estimate is more than 40\% worse than Direct. The \R1 sd is somewhat greater than Direct at level $b$ but still a considerable improvement over Classical.

The surprise here is that using the unlabeled data does \textit{not} improve the estimation of $\theta=E\{y\}$. It is true, as emphasized in the literature, that $\hthe_b$ is less variable than the classical estimator mean$(\by_a)$ but it isn't better than $\hthe_a$, the estimate based just on $\bx_a$, $\by_a$, and $\bmf_a$. To put it another way, there is useful information about $\theta$ in $\bmf_a$ but not in $\bmf_b$.

The example in \ref{tab5} holds generally: for model$_a$ \eqref{26}, \textit{prediction-powered inference does not improve the estimation of} $\theta=E\{y\}$. A supporting argument follows next, based on \eqref{219} and \eqref{221}. Note: the fact that the classical sd estimate 0.026 is a factor of about $\sqrt2$ greater than the direct sd estimate 0.018 is not a coincidence; see \ref{rem10} at the end of this section.

Formula \eqref{219} can be written as
\begin{equation}
\sd\left(\hthe_a\right)=M_aG^{-1}M_a'\qquad\left(M_a=\hbdel_a'V_aL_a\right),
\label{43}
\end{equation}
$L_a=(\bone,\bl_a)$ as before. For $\theta=E\{y\}$, the gradient vector \eqref{218} is
\begin{equation}
\hbdel_a=\left(\frac{\bone_a}{n_a}\right)\qquad(\bone_a\text{ the vector of $n_a$ 1s}),
\label{44}
\end{equation}
giving
\begin{equation}
M_a=\frac1{n_a}\left[\sum_{i=1}^{n_a}\pi_a(i)(1-\pi_a(i)),\sum_{i=1}^{n_a}\pi_a(i)(1-\pi_a(i))\logit\pi_a(i)\right].
\label{45}
\end{equation}
Similarly from \eqref{221},
\begin{equation}
\sd\left(\hthe_b\right)\doteq M_bG^{-1}M_b',
\label{46}
\end{equation}
where
\begin{equation}
M_b=\frac1{n_b}\left[\sum_{i=1}^{n_b}\pi_b(i)(1-\pi_b(i)),\sum_{i=1}^{n_b}\pi_b(i)(1-\pi_b(i))\logit\pi_b(i)\right].
\label{47}
\end{equation}

Defining the functions
\begin{subequations}\begin{align}
\gamma(\pi)=\pi(1-\pi)\label{48a}\\
\intertext{and}
\delta(\pi)=\pi(1-\pi)\logit\{\pi\},\label{48b}
\end{align}\label{48}
\end{subequations}
$M_a$ \eqref{45} can be written as
\begin{equation}
M_a=\left(\hate_a\{\gamma(\pi)\},\hate_a\{\delta(\pi)\}\right),
\label{49}
\end{equation}
where $\hate_a$ indicates expectation with respect to the distribution putting probability $1/n_a$ on each point $\pi_a(i)$. Likewise,
\begin{equation}
M_b=\left(\hate_b\{\gamma(\pi)\},\hate_b\{\delta(\pi)\}\right),
\label{410}
\end{equation}
$\hate_b$ the expectation with probability on each point $\pi_b(i)$.

\ref{rem5} of \ref{sec3} says that the distribution of $\hpi_b$ matches that of $\hpi_a$, making $M_b=M_a$; so that $\sd(\hthe_a)\doteq\sd(\hthe_b)$, as claimed.

\begin{table}[htbp]
\caption{\R1 applied to the standard deviation, skewness, and kurtosis of $\by$, Pew0 data. For all three statistics, prediction-powered variability $\hsd_b$ exceeds $\hsd_a$ (that based on only labeled data), although $\hsd_b$ is less than the nonparametric bootstrap standard deviation based on $\by_a$ (``classical'').}
\begin{center}\begin{small}\begin{tabular}{lccc}
    &$\sd_a$&$\sd_b$&Classical\\\hline
Sd  &.009&.011&.013\\
Skew&.097&.119&.148\\
Kurt&.202&.259&.314 
\end{tabular}\end{small}\end{center}
\label{tab6}
\end{table}

\ref{tab6} concerns the application of \R1 to three other statistics that depend only on $\by$: the standard deviation, skewness, and kurtosis of the $n$ values $y(i)$. Once again the PPI estimate $\hthe_b$ is seen to be more variable than $\hthe_a$ (though $\hsd_b$ \textit{is} less than the ``classical'' nonparametric bootstrap standard deviation estimate based on $\by_a$). A reasonable conjecture is that $\hthe_b$ is never more accurate than $\hthe_a$ when $t(\bx,\by)$ depends only on $\by$. See the discussion in \ref{rem11} and also that in \ref{sec5}.

The results in \ref{fig3} are more encouraging, demonstrating considerable efficacy for PPI estimation of correlation$(\bx,\by)$ for the Pew0 data. There, $n_a=300$ and $n_b=900$. \ref{fig8} shows that increasing $n_b$ to 1800 and then 3600 gives still better results.

\begin{figure}[htbp]
\centering
\includegraphics{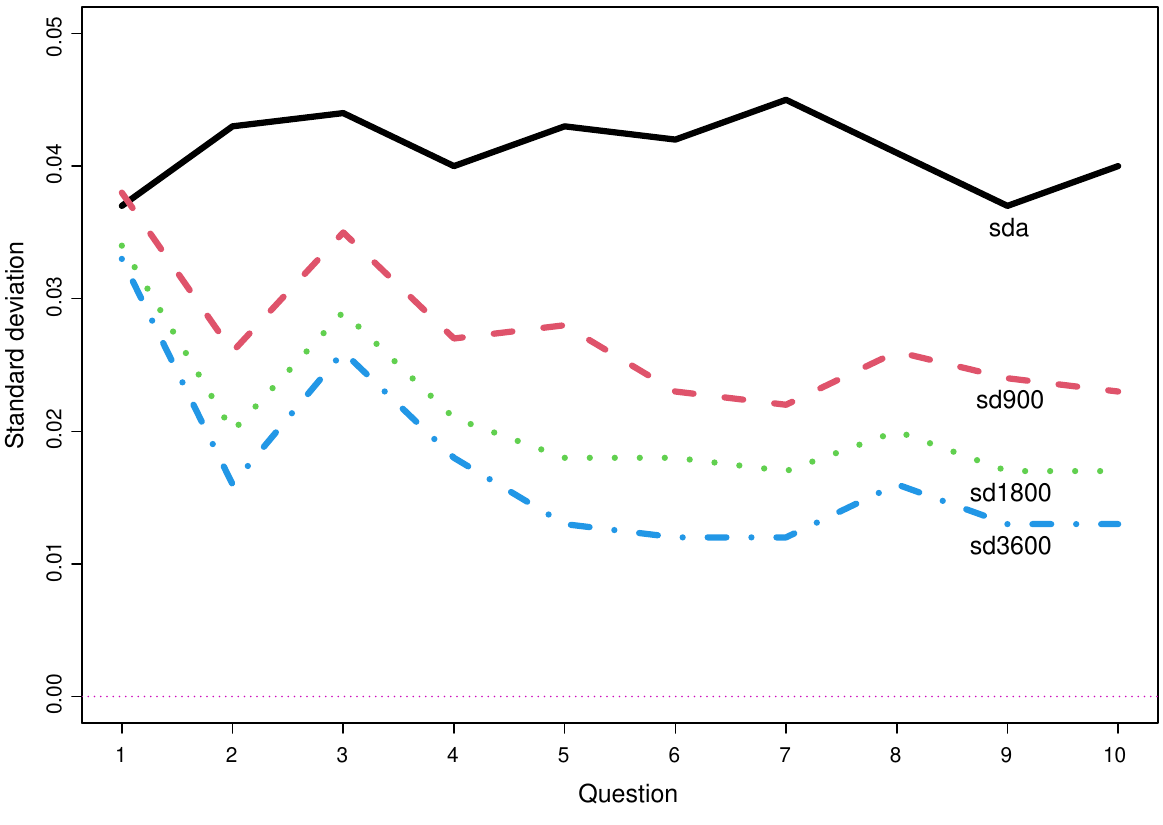}
\caption{\R{1} estimated standard deviations, $t(x,y)=\corr(x,y)$, Pew Center data; $n_a=300$ and $n_b=900$, 1800, 3600.}
\label{fig8}
\end{figure}

Looking at \ref{tab6} and \ref{fig8} raises an intriguing theoretical question: given a particular dataset, statistic $t(\bx,\by)$, $n_a$, and $n_b$, how much can we expect to gain from prediction-powered inference? At the opposite extreme from parameters like $E\{y\}$ which only depend on $\by$, suppose $t(\bx,\by)$ is a function of $\bx$ alone. Then $\hsd_b$ will decline as $n_b^{-1/2}$. (In fact, sd3600 in \ref{fig8} \textit{is} about half of sd900, but $\corr(\bx,\by)$ is not necessarily a typical case.) The expressions \eqref{219} and \eqref{221} for $\hsd_a$ and $\hsd_b$ seem like they might answer the intriguing question. But, as before, this is misleading: in most cases, the gradients \eqref{218} and \eqref{223} are not readily available. See \ref{rem10} at the end of this section.

There is some good news in \ref{tab4}: when an effective prediction rule $f$ is available, prediction-powered methods can aid in analyzing the labeled data even in the absence of unlabeled data. A small example follows.

\begin{table}[htbp]
\caption{First 5 of 130 subjects in sick babies dataset. Six covariates used to predict death ($y=1$) or survival ($y=0$): gestational age, body weight index, respiratory efficiency, breathing measure, mental fitness, heart rate (all in standardized units).}
\begin{center}\begin{small}\begin{tabular}{rrrrrrc}
   Gest&  Bwei&   Resp&  Cpap&  Ment&   Rate&$y$\\\hline 
$  .91$&$ .87$&$  .56$&$-.71$&$1.45$& $-.44$& $0$\\
$-1.02$&$-.46$&$ 1.41$&$1.40$&$1.45$&$-1.87$& $1$\\
$ 1.15$&$ .87$&$ 1.41$&$1.40$&$1.45$&$-1.87$& $1$\\
$  .18$&$-.46$&$-1.13$&$-.71$&$-.69$& $-.44$& $1$\\
$  .18$&$-.46$&$-1.13$&$-.71$&$-.69$& $-.44$& $0$
\end{tabular}\end{small}\end{center}
\label{tab7}
\end{table}

\ref{tab7} shows data for the first five of 130 subjects in a study of very sick babies at an African hospital. Each child either died within a few weeks after arrival ($y=1$) or survived ($y=0$); 52\% of the babies died. Six baseline covariates are available for predicting $y$. A larger study was carried out in the previous year, involving 812 babies measured on the same six covariates. However, looser criteria admitted greater numbers of healthy babies to the study, and this resulted in a lower 26\% death rate. Can we use the first study to help with estimation in the second one?

Consider the first study as ``background data'' and the small second study as the ``labeled data''. A standard logistic regression of the background dataset yielded a vector of regression coefficients (length 7 including intercept); these were applied to the labeled set to obtain an $\bmf_a$. We now have $\bx_a$, $\by_a$, and $\bmf_a$, and can carry out the part of \R1 that applies to the labeled data.

This was performed for $\theta=E\{y\}$ with these results:
\begin{center}\begin{small}\begin{tabular}{rcc}
&\R1&Classical\\\hline
$\hthe$&.522&.523\\
$\hsd$&.034&.044
\end{tabular}\end{small}\end{center}

\noindent The \R1 estimate is nearly identical to the Classical mean ($\by_a$) but has confidence interval only 76\% as long.

\begin{rem}
The ``Direct'' standard deviation estimates in \ref{tab5} refer to conditional resampling $\bya_a\sim\bern(\hbpi_a)$. If instead we used unconditional resampling with $\bya_a$ nonparametrically resampled from $\by_a$, the direct $\hsd_a$ estimate would be identical to the classical sd.
\label{rem9}
\end{rem}

\begin{rem}
The ratio of classical to direct variances in \ref{tab4} is $(0.026/0.018)^2=2.09$. For estimating $E\{y\}$, it turns out that the ratio
\begin{equation}
R=\left(\frac{\hsd_{\cla}}{\hsd_a}\right)^2
\label{411}
\end{equation}
generally has $R\doteq2.0$, as the following argument shows.

In step 1 of \ref{fig1}, $\bya_a\sim\bern(\hbpi_a)$ with $\hbpi_a$ as in \eqref{28}. Letting
\begin{equation}\begin{aligned}
\theta_a&=\frac1{n_a}\sum_1^{n_a}\pi_a(i),\\
\hthe_a&=\frac1{n_a}\sum_1^{n_a}\hpi_a(i),\\
\barya_a&=\frac1{n_a}\sum_1^{n_a}y_a(i)^*,
\end{aligned}\label{412}
\end{equation}
and writing
\begin{equation}
\barya_a-\theta_a=\left(\hthe_a-\theta_a\right)+\left(\barya_a-\hthe_a\right)
\label{413}
\end{equation}
gives an expression for the marginal variance of $\barya$,
\begin{equation}
\var(\barya_a)=\var\left(\hthe_a\right)+\var(\barya_a\mid\hbpi_a).
\label{414}
\end{equation}

We have, from \eqref{219} and \eqref{41},
\begin{equation}
\var\left(\hthe_a\right)\doteq\bv'L_aG^{-1}L_a'\bv/n_a^2\qquad\left(v(i)=\hpi_a(i)(1-\hpi_a(i))\right)
\label{415}
\end{equation}
and
\begin{equation}
\var(\barya_a\mid\hbpi_a)=\frac1{n_a^2}\sum_1^{n_a}v(i).
\label{416}
\end{equation}
But, as shown in the following lemma, \eqref{415} equals \eqref{416}, so
\begin{equation}
\var(\barya_a)\doteq2\var\left(\hthe_a\right).
\label{417}
\end{equation}
Since $\var(\barya_a)$ equals, approximately, $\sd^2_{\cla}$, this shows that $R\doteq2$. (If $\hbpi_a$ is unbiased for $\bpi_a$ then $\var(\barya_a)$ will equal $\var(\bary_a)$; the difference $\var(\bary_a)-\var(\barya_a)$ is asymptotically negligible.) Recalculating \ref{tab4} for the 100 datasets in Pew1:100 gave average $R$ value 2.16 with standard deviation 0.27.
\label{rem10}
\end{rem}

\begin{lem}
The right-hand sides of \eqref{415} and \eqref{416} are the same.
\end{lem}

\begin{proof}
Letting $\till=V^{1/2}L_a$ with $V=\diag(\bv)$, and $\tbv=\bv^{1/2}$, \eqref{415} becomes
\begin{equation}
\var\left(\hthe_a\right)\doteq\frac{\tbv'\till\left(\till'\till\right)^{-1}\bL\tbv}{n_a^2}.
\label{418}
\end{equation}
The right-hand side is the squared length of projection of $\tbf_a$ into $\callcol(\till_a)$, the column space of $\till_a$. But
\begin{equation}
\till=V^{1/2}(\bone,\bl_a)=(\bv^{1/2}V^{1/2}\bl_a),
\label{419}
\end{equation}
so $\tbv$ is in $\callcol(\till)$. Therefore, equating the lengths of projection,
\begin{equation}
\frac1{n_a^2}\bv_a'L_aG^{-1}L_a'\bv_a=\frac1{n_a^2}\sum_{i=1}^{n_a}v(i).
\label{420}
\end{equation}
\end{proof}

\begin{figure}[htbp]
\centering
\includegraphics{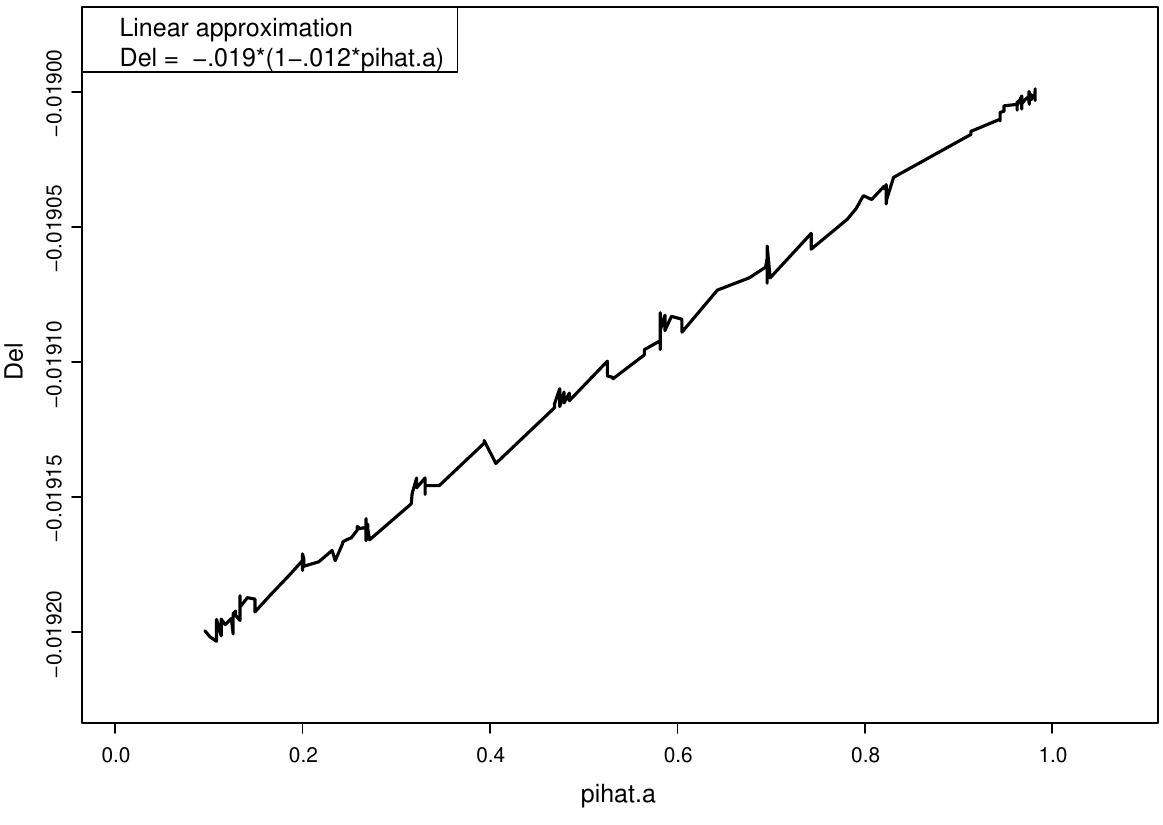}
\caption{Estimated gradient $\bdel$ for skewness statistic; Pew0 data, as a function of $\hpi_a$.}
\label{fig9}
\end{figure}

\begin{rem}
Aside from the mean \eqref{41}, most parameters $\theta=T(\bpi)$ don't enjoy simple formulas for the gradient vector $\bdel=dT(\bpi)/d\bpi$. However, the bootstrap computations in \R1 can be reverse-engineered to estimate $\bdel_a$ \eqref{218} and $\hbdel_b$ \eqref{223} for statistics $t(\by)$ \eqref{210} that do not depend on $\bx$.

For a given $i$ in $1,\dots,n_a$, the expectation $\hthe_a=\eepa\{t(\by_a)\}$ can be expressed as
\begin{equation}
\hthe_a=\hpi_a(i)\eepa\{t(\by_a)\mid y_a(i)=1\}+(1-\hpi_a(i))\eepa\{t(\by_a)\mid y_a(i)=0\},
\label{421}
\end{equation}
so
\begin{equation}\begin{aligned}
\hdel_a(i)&=\frac{\partial T(\hbpi_a)}{\partial\hpi_a(i)}\\
          &=\eepa\{t(\by_a)\mid y_a(i)=1\}-\eepa\{t(\by_a)\mid y_a(i)=0\}.
\end{aligned}\label{422}
\end{equation}

Let $\bY$ be the $B\times n_a$ matrix having the $j$th bootstrap replication $\bya_a$ (step 1 in \ref{fig4}) as its $j$th row; $\bY_1(i)$ the version of $\bY$ that has column $i$ changed to all 1s; and $\bY_0(i)$ with the $i$th column changed to all 0s. Define $\hthe_1(i,j)$ as the estimate $t(\by_a)$ computed from the $j$th row of $\bY_1(i)$, and similarly $\hthe_0(i,j)$ from $\bY_0(i)$. Averaging $\hthe_1(i,j)-\hthe_0(i,j)$ over $j=1,\dots,B$ provides an estimate of $\Delta_a(i)$,
\begin{equation}
\hdel_a(i)=\frac1{B}\sum_{j=1}^B\left[\hthe_1(i,j)-\hthe_0(i,j)\right].
\label{423}
\end{equation}

Formula \eqref{423} was applied to the response vector $\by_a$ of the Pew0 data with the statistic $t(\by_a)$ the empirical skewness. \ref{fig9} graphs $\hdel_a(i)$ versus $\hpi_a(i)$, showing a slightly noisy linear increase. This is misleading: a linear fit to the graph gives
\begin{equation}
\hbdel_a\doteq-0.019\cdot(1-0.012\cdot\hbpi_a).
\label{424}
\end{equation}
That is, $\hbdel_a$ is nearly flat, varying by only 1.2\% as $\hbpi_a$ goes from 0 to 1. This isn't the same as the mean \eqref{41} where $\bdel_a$ is perfectly flat, but it's close.

Gradient vector $\hbdel_b$ \eqref{223} also can be estimated from formula \eqref{423}, now applied to the bootstrap replications $\bya_b$, step 4 in \ref{fig4}. For the Pew0 skewness example the calculations show
\begin{equation}
\hbdel_b\doteq\hbdel_a/3;
\label{425}
\end{equation}
the factor 3 follows from $n_b=3n_a$. Having $\hbdel_b\doteq\hbdel_a\cdot(n_a/n_b)$ results in $\hsd_a\doteq\hsd_b$, from the same argument in \eqref{43} through \eqref{410}. All of this supports the conjecture that, in the binary case, $\hthe_b$ is no improvement on $\hthe_a$ for statistics $t(\bx,\by)$ that depend only on $\by$.
\label{rem11}
\end{rem}

\section{Quantitative response}\label{sec5}

Quantitative data, where the response variable $y$ is real-valued rather than binary, presents a greater challenge for prediction-powered inference. The reason is simple: there is only one binary distribution \eqref{22} but a host of quantitative possibilities. This section discusses a quantitative-response version of model$_a$ \eqref{26} and an algorithm \R2 for carrying out the PPI calculations. It has close connections with the pioneering methods of \cite{wang}. \ref{rem12} say a little more about the relationship.

\begin{table}[htbp]
\caption{Five of 400 subjects in labeled dataset of Cens0. Shown are 5 of 8 predictor variables, $y_a$ the amount of tax paid and $f_a$ the machine learning prediction.}
\begin{center}\begin{small}\begin{tabular}{cccccrc}
AGEP&SCHL&MAR&SEX&RAC1P& $y_a$&$f_a$\\\hline
  58&  17&  1&  1&    3& 43700&38062\\
  65&  16&  3&  1&    1& 10400&17862\\
  68&  22&  2&  1&    1&241900&78775\\
  67&  20&  1&  2&    1& 10290&23968\\
  42&  12&  1&  1&    8& 40000&34329
\end{tabular}\end{small}\end{center}
\label{tab8}
\end{table}

\ref{tab8} shows a small part of the Census Data, an example featured in this section: 160,000 California residents reported their income $y$ (truncated at 250,000) and eight demographic prediction variables $x$. The data from 80,000 of them were used to calculate a machine learning prediction function $f(x)$; the remaining 80,000 were used to provide randomly selected sets of labeled and unlabeled data \eqref{11}--\eqref{12}. Cens0, my main example, has a labeled dataset of size $n_a=400$ and an unlabeled set of $n_b=1200$ cases ($y$ values erased). Cens1:100 comprises 100 additional datasets, each with $n_a=400$ and $n_b=1200$ members.

The quantitative analogue of the binary response model$_a$ \eqref{26} is model$_q$,
\begin{equation}
y_a(i)=\mu_a(i)+\sigma_a(i)\cdot\epsilon(i)\qquad i=1,\dots,n_a,
\label{51}
\end{equation}
a heteroskedastic linear regression model. Here $\mu_a(i)$ and $\sigma_a(i)$ are real-valued functions $m\pdot$ and $s\pdot$ of $f_a(i)$,
\begin{equation}
\mu_a(i)=m(f_a(i))\quand\sigma_a(i)=s(f_a(i)),
\label{52}
\end{equation}
while the $\epsilon(i)$ are iid draws from a given probability density $p$ having mean 0 and variance 1,
\begin{equation}
\epsilon(i)\iid p\pdot,\qquad i=1,\dots,n_a.
\label{53}
\end{equation}
The following discussion assumes conditional inference where $\bx_a$ and $\bmf_a$ are considered fixed.

\begin{figure}[htbp]
\centering
\includegraphics{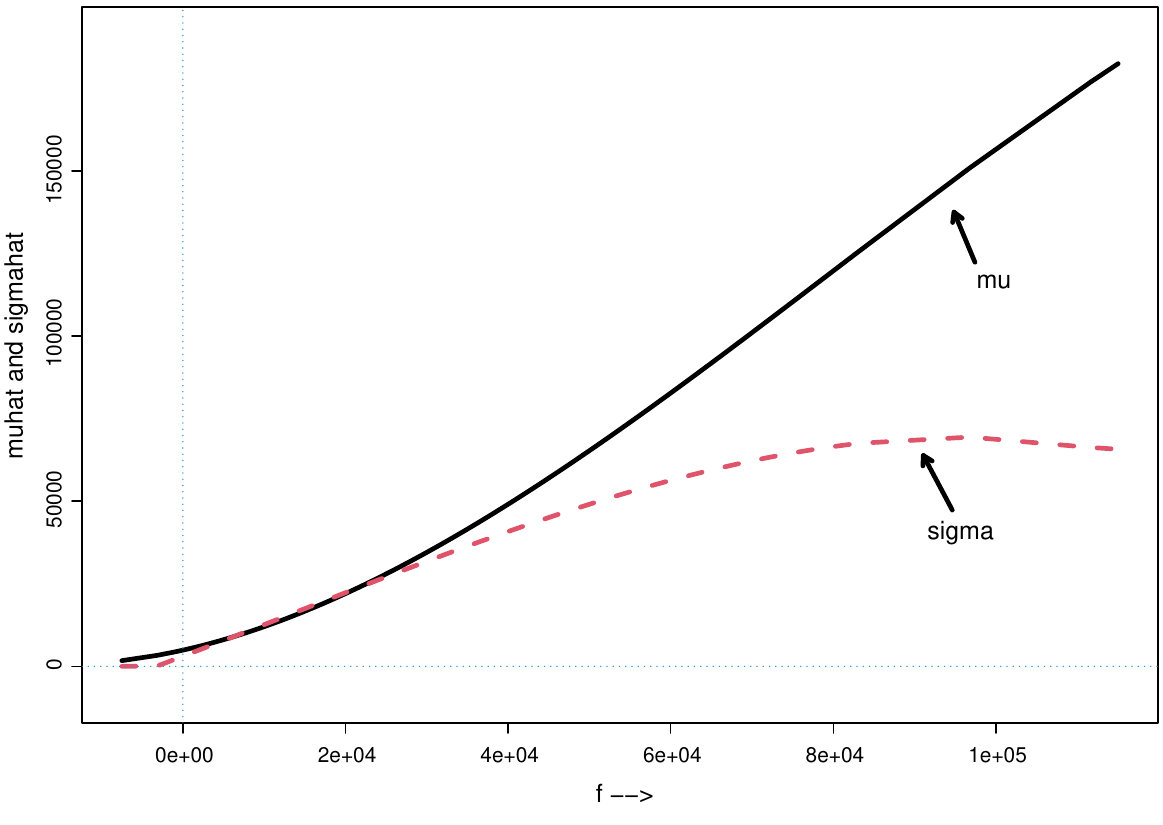}
\caption{Estimated mean $\mu=m(f)$ and standard deviation $\sigma=s(f)$ as a function of $f$; Cens0 labeled data.}
\label{fig10}
\end{figure}

\R2 uses low-dimensional least squares regression models on the labeled data to estimate the function $m$ and successive differences of the responses $y_a(i)$ to estimate $s$, as described in \ref{rem13}. \ref{fig10} shows the estimates $\hatm(f)$ and $\hats(f)$ based on the 400 values of $y_a(i)$ in the labeled Cens0 data. The ratio $m(f)/f$ increases from 1.0 to about 1.5 for large values of $f$.

Model$_a$ assumed $\by_a\sim\bern(\bpi_a)$ \eqref{26}. A more complicated definition of the generative distribution $\bpi_a$ is required for model$_q$:
\begin{align}
\bpi_a&=(\bmu_a,\bsig_a,p),\label{54}\\
\intertext{where}
\bmu_a&=(\cdots\mu_a(i)\cdots),\notag\\
\bsig_a&=(\cdots\sigma_a(i)\cdots),\notag
\end{align}
and $p$ is the probability distribution \eqref{53} for the errors $\epsilon(i)$. At step 2 of \ref{fig2}, $\hbet$ is now defined as
\begin{equation}
\hbet=(\hatm,\hats,\hro),
\label{55}
\end{equation}
$\hatm,\hats$ and $\hro$ being estimates of $m,s$, and $p$ based on $\bx_a,\by_a$ and $\bmf_a$. (See \ref{rem13} for details.)

With this change, \ref{fig2} and \ref{fig4} describe the logic of \R2 as well as \R1. For instance,
\begin{equation}
\hbpi_a=(\hbmu_a=\hatm(\bmf_a),\hbsig_a=\hats(\bmf_a),\hro)
\label{56}
\end{equation}
generates
\begin{equation}
y_a^*(i)=\hmu_a(i)+\hsig_a(i)\cdot\hep_i,\qquad i=1,\dots,n_a,
\label{57}
\end{equation}
at step 1 of \ref{fig4} with $\hep_i\iid\hro$. Similarly,
\begin{equation}
\hbpi_b=(\hbmu_b=\hatm(\bmf_b),\hbsig_b=\hats(\bmf_b),\hro)
\label{58}
\end{equation}
gives
\begin{equation}
y_b^*(i)=\hmu_b(i)+\hsig_b(i)\cdot\hep_i
\label{59}
\end{equation}
at step 4 of \ref{fig4} in \ref{sec3}.

\begin{figure}[htbp]
\centering
\includegraphics{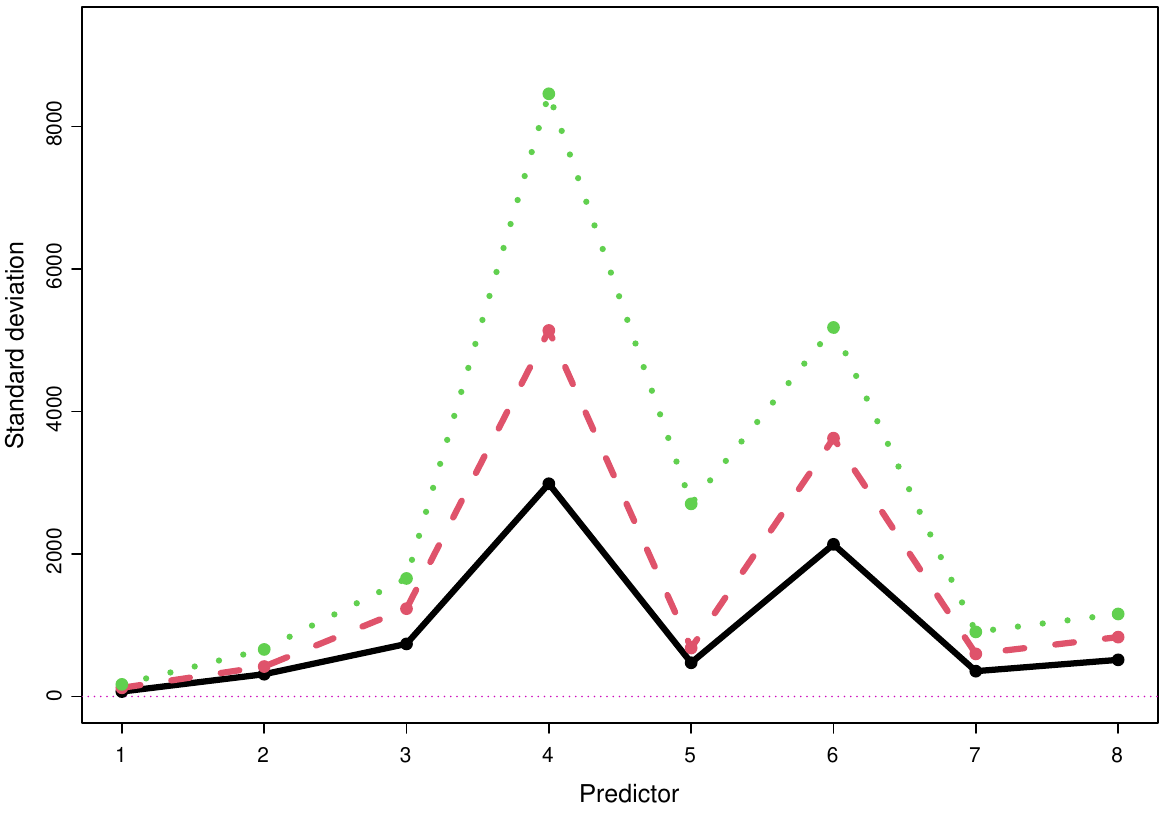}
\caption{Standard deviations of regression coefficients, Cens0 data; labeled and unlabeled (solid), labeled only (dashed), classical (dotted).}
\label{fig11}
\end{figure}

\R2 was applied to the estimation of the linear model regression coefficients for the Cens0 dataset. Standard deviations of the point estimates are graphed for three estimators in \ref{fig11}:
\begin{itemize}
\item Classical: the usual least squares variance formula for the regression coefficients from $y_a\sim x_a$.
\item Labeled data: from \R2 applied to $\bx_a$, $\by_a$, $\bmf_a$.
\item Labeled and unlabeled: \R2 applied to all the Cens0 data.
\end{itemize}

\noindent The estimates are conditional with $\bx_a$ and $\bx_b$ considered fixed.

The story is almost the same as in \ref{fig3}: using all the data, labeled and unlabeled reduces standard deviations by about 60\% compared to the classical estimates. Using just the labeled data gives about half as much reduction in variability. As far as means are concerned, the three methods had similar expectations with no significant differences observed.

\begin{figure}[htbp]
\centering
\includegraphics{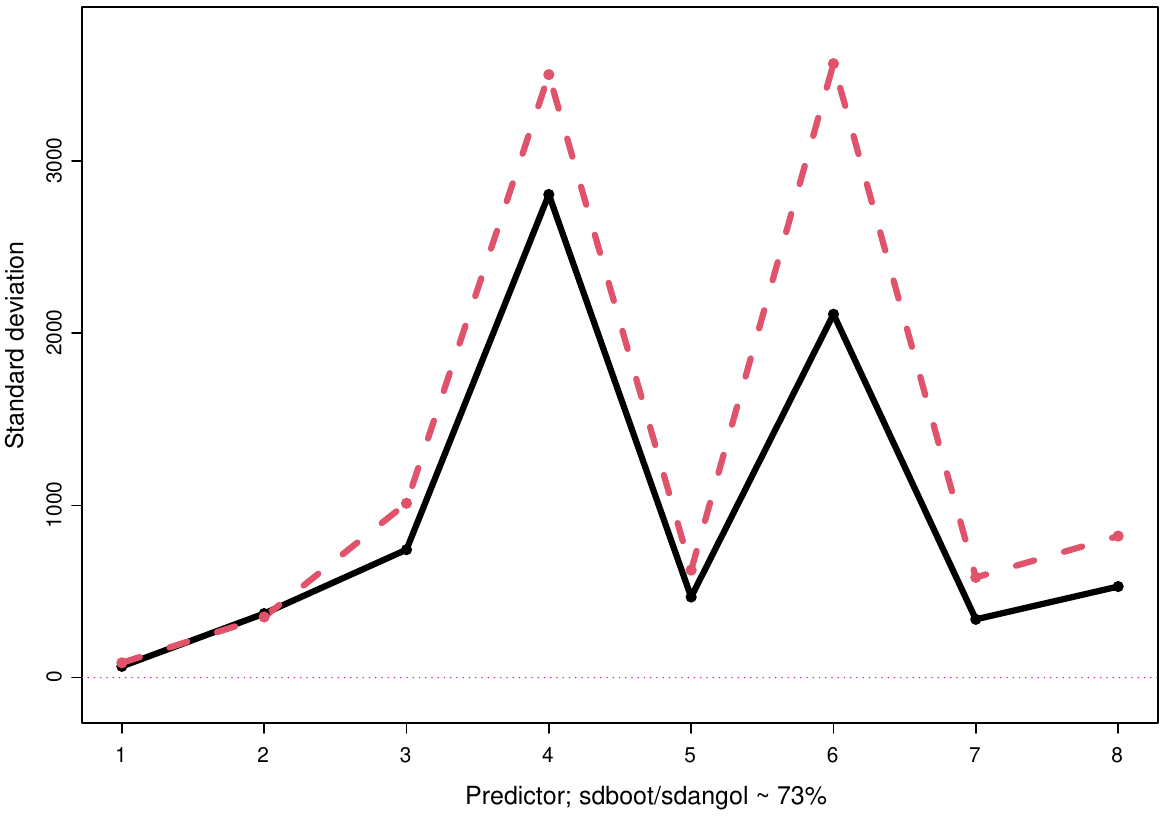}
\caption{Comparison of sds for Cens0 regression coefficients; \R{2} (solid) and \cite{angel} (dashed).}
\label{fig12}
\end{figure}

\ref{fig12} compares standard deviations of the \R2 estimates (the solid curve in \ref{fig11}) with those obtained from the algorithm of \cite{angelx}. Over the eight predictors, the \R2 sds are about 74\% as large. The Angelopoulos calculations proceeded as in \eqref{35}--\eqref{36}, with $D=\|\by-\bx\hthe\|^2$. They perform better here than in \ref{fig6}, with about 34\% increased sd compared to \R2, which is perhaps a reasonable price to pay for not requiring the assumptions of model$_q$. (\R2 was run using unconditional resampling at the first level, step 1 in \ref{fig4}, and conditional resampling at step 4.)

\begin{figure}[htbp]
\centering
\includegraphics{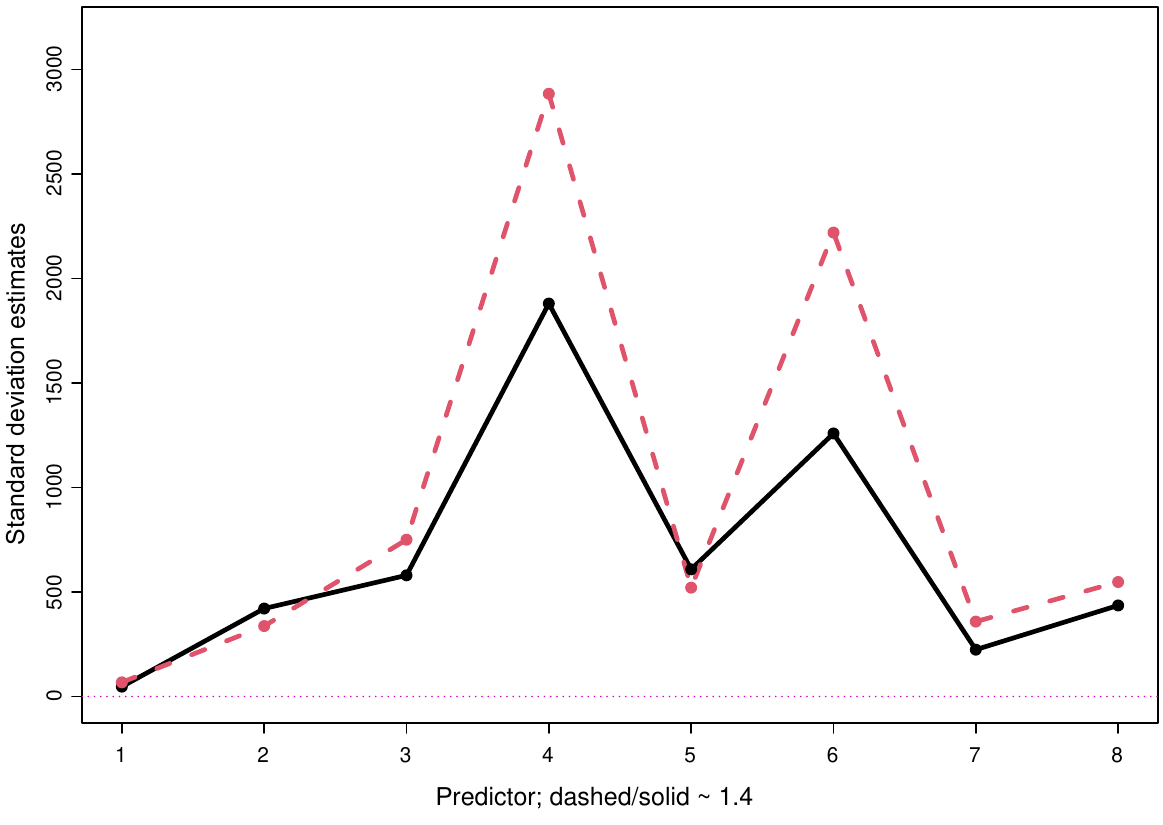}
\caption{\R{2} applied to 100 datasets Cens1:100; sd of $\hthe_b$ estimates (solid), average of estimated sds (dashed).}
\label{fig13}
\end{figure}

\ref{fig13} concerns the legitimacy of the bootstrap confidence intervals from \R2: are they wide enough? To put it negatively, are the standard deviation estimates $\hsd_b$, shown as the solid curves in \ref{fig11} and \ref{fig12}, too small? As an answer, the simulation test of \ref{fig5} was repeated for the Census Data regression coefficients. \R2 was run for each of the 100 datasets Cens1:100, each with $n_a=400$ and $n_b=1200$. The solid curve in \ref{fig13} tracks the external standard deviation \eqref{312}; that is, the empirical sd of the 100 estimates $\hthe_b$ (for each of the eight predictors). The dashed curve is the internal standard deviation \eqref{313}, i.e., the average of the 100 $\hsd_b$ estimates.

As in \ref{fig5}, \ref{fig13} suggests that the bootstrap estimates $\hsd_b$ exceed the true standard deviations by about 40\% for the eight predictors. The figure was recalculated using only the labeled data, that is, for $\hthe_a$ and $\hsd_a$, as in \ref{rem8} of \ref{sec3}. In this case, however, the recalculated figure looked almost the same as \ref{fig13}, still indicating substantial overcoverage.

\begin{table}[htbp]
\caption{Average \R2 estimates $\sd_a$ and $\sd_b$ of the mean, sd, skewness, and kurtosis of the distribution of $y$, for the 100 datasets Cens1:100.}
\begin{center}\begin{small}\begin{tabular}{lccll}
       &Mean&Sd&  Skew&Kurt\\\hline
$\sd_b$&2070&1871&.096&.471\\
$\sd_a$&1837&2506&.126&.612\\
\\
ratio  &1.13& .75&.77 &.78
\end{tabular}\end{small}\end{center}
\label{tab9}
\end{table}

\R2 was applied to the Cens0 data for four functions of the responses $y$: mean, standard deviation, skewness, and kurtosis. \ref{tab9} provides standard deviations for the two methods and their ratio. Unlike the binary model results in \ref{tab6}, using all the data, labeled and unlabeled, performs better than using labeled alone for Sd, Skew, and Kurt, but not for Mean. Once again, prediction-powered inference seems to be ineffective for estimating the expectation of $y$, as discussed next.

The binary response formulas \eqref{219} and \eqref{221} gave approximations $\hsd_a$ and $\hsd_b$ for $\hthe_a$ and $\hthe_b$. Similar results apply in the quantitative response case. Now let $L_a$ be a matrix of dimension $n_a$ by df$+1$,
\begin{equation}
L=(\bone,\text{poly}(\bmf,\text{df}),
\label{510}
\end{equation}
the R notation indicating an orthogonal polynomial basis of degree df in the elements of $\bmf_a$. \R2 uses a $C_p$ criterion to choose the most predictive value of df between 1 and 5 (\ref{rem13}).

Under model \eqref{54} the OLS covariance matrix for the estimate $\hbmu_a$ is
\begin{equation}
\cov(\hbmu_a)=L_aG^{-1}L_a',
\label{511}
\end{equation}
where
\begin{equation}
G=L_a'\diag(\bs^{-2})L_a,
\label{512}
\end{equation}
the notation indicating a diagonal matrix with entries $s(f_a(i))^{-2}$.

Suppose the parameter of interest $\theta_a$ is a function of $\bmu_a$, say,
\begin{align}
\theta_a&=T(\bmu_a),\notag\\
\intertext{for instance}
T(\bmu_a)&=\frac{\sum_1^{n_a}\mu_a(i)}{n_a}\for\theta=E\{y\}.\label{513}
\end{align}
Let $\bdel_a$ be the gradient vector
\begin{equation}
\bdel_a=\left.\frac{dT(\bmu)}{d\bmu}\right|_{\bmu_a}.
\label{514}
\end{equation}
The delta method approximation for the standard deviation of $\hthe_a=T(\hbmu_a)$ is
\begin{equation}
\hsd_a=\left[\bdel_a'L_aG^{-1}L_a'\bdel_a\right]^{1/2}.
\label{515}
\end{equation}
Similarly, the estimated standard deviation for $\hthe_b=T(\hbmu_b)$ is
\begin{align}
\hsd_b&=\left[\bdel_b'L_bG^{-1}L_b'\bdel_b\right]^{1/2},\label{516}\\
\intertext{where}
L_b&=(\bone_b,\text{poly}(\bmf_b,\text{df})),\notag
\end{align}
$\bone_b$ the vector of $n_b$ 1s.

The expectation parameter $\theta_a=E\{\by_a\}$ corresponds to \eqref{513}, for which $\bdel_a=\bone/n_a$. Because poly($\bmf_a$,df) is orthogonal to $\bone_a$ in \eqref{510}, we get
\begin{equation}
\bdel_a'L_a=(1,\bzer)
\label{517}
\end{equation}
and \eqref{515} becomes
\begin{equation}
\hsd_a=[G^{-1}]_{11}^{1/2}.
\label{518}
\end{equation}

The same argument gives $\bdel_b=\bone_b/n_b$ and
\begin{equation}
\hsd_b=[G^{-1}]_{11}^{1/2},
\label{518x}
\end{equation}
equaling $\hsd_a$ and strengthening the conclusion that unlabeled data doesn't help in the estimation of $E\{y\}$.

\section*{Bias}
In \ref{fig2}, the true generative model $\bpi_a$ produces an estimated model $\hbpi_b$ and a point estimate $\hthe_b=T(\hbpi_b)$ for $\theta=T(\bpi_b)$. The function $T\pdot$ is smoothly defined because of the expectation step $T(\bpi)=E_{\bpi}\{t(\bx,\by)\}$, making it plausible that $\hthe_b$ is nearly unbiased for $\theta_b$, in the sense of maximum likelihood, where bias decreases an order of magnitude faster than standard deviation.

The relevance of this argument depends on $T(\bpi)=E_{\bpi}\{t(\bx,\by)\}$ being the actual parameter of interest. As an unfavorable example, suppose $t(\bx,\by)$ is the usual estimate of $R^2$ from the linear model $\by\sim\bx$. Applying \R2 to the Cens0 data gave $T(\hbpi_a)=0.253$ with standard deviation $\hsd_a=0.033$. This was 0.299 standard deviations above the classical estimate $t(\bx,\by)=0.241$ and indicates a substantial upward bias. (\R2 used nonparametric bootstrapping at step 1 of \ref{fig4}, so these calculations didn't involve the modeling assumptions \eqref{51}.)

The bootstrap confidence interval algorithm \texttt{bc} (\texttt{bca} with $a=0$) of \cite{1987} suggests adding the bias corrector
\begin{equation}
\gamma=\hsd_b\cdot z_0
\label{520}
\end{equation}
to the endpoints of the uncorrected interval $0.253\pm0.033$, where
\begin{equation}
z_0=\Phi^{-1}(P_0),
\label{521}
\end{equation}
$\Phi$ is the standard normal cdf, and $P_0$ is the proportion of bootstrap replicates less than $t(\bx,\by)$; $P_0=0.383$ here, giving $\gamma=-0.010$. The uncorrected 95\% interval $(0.188,0.318)$ becomes $(0.178,0.308)$ after the bias correction. Bias corrector \eqref{520} is always available from the output of \R2 as a warning of potential bias problems.

\begin{rem}
The development in this section is closely related to the seminal paper of \cite{wang}, henceforth WML. Here is a brief review of the connections. WML assume a three-stage model:
\begin{enumerate}
\item A training set of $(x,y)$ pairs provides a \textit{prediction model} $f$ for response $y$ given predictor $x$,
\begin{equation}
\haty=f(x).
\label{522}
\end{equation}
\item A testing set of $(x,y)$ pairs compares predictions $\haty=f(x)$ with the observed $y$ values, from which a \textit{relationship model} $k\pdot$ giving improved predictions,
\begin{equation}
\ypre=k(\haty),
\label{523}
\end{equation}
is constructed.
\item A validation set of $x$ values substitutes $\ypre$ for the missing $y$ values; the pairs $(x,\ypre)$ are used as if they were $(x,y)$ pairs to calculate the parameter of interest.
\end{enumerate}

The training, testing, and validation sets of WML correspond to the background data \eqref{14}, labeled data \eqref{11}, and unlabeled data \eqref{12} of this paper.

There is an important conceptual difference: in this paper the machine learning function $f(x)$ is taken to estimate the expected value $E\{y\mid x\}$, not $y$ itself. This plays a central role in the bootstrap algorithm of \ref{fig4}. The bootstrap-based correction of WML (p.~30269) resamples only at the validation level, that is, only on the unlabeled data. This seems to ignore sampling variability in the form of the relationship function $k\pdot$. The algorithm here requires \textit{two} resampling steps, at steps 1 and 4 in \ref{fig4}, which adds to the estimated variability of $\hthe_b$.
\label{rem12}
\end{rem}

\begin{rem}
\R2 is more speculative than \R1 because model$_q$ \eqref{51}--\eqref{53} is less specific than model$_a$ \eqref{26}. The following choices were made for the version of \R2 used in the previous calculations:
\begin{itemize}
\item $\hatm(f)$ is fit by least squares to $\by$ as a polynomial function of $\bmf$; if not specified, the degree df ($\leq5$) is selected by $C_p$ minimization.
\item $\hats(f)$ is a smoothed version of the successive ordered absolute differences
\begin{equation}
0.886\cdot|y(i+1)-y(i)|;
\label{524}
\end{equation}
the constant 0.886 gives unbiased estimates of the local standard deviation if the $y$s are independently normal.
\item The residual density $p\pdot$ is taken to be a standardized Gamma distribution with $\nu$ degrees of freedom, $\nu$ chosen by a robustified fit to the standardized residuals
\begin{equation}
\frac{y(i)-m(i)}{s(i)}.
\label{525}
\end{equation}
For the examples in this section, \R2 chose df $=4$ and $\gamma=0.924$. Negative values of $\nu$ give reversed Gamma residuals, long-tailed to the left.
\end{itemize}
\label{rem13}
\end{rem}

\section*{Acknowledgments}
I am very grateful to Tijana Zrnic at Stanford for many helpful discussions as well as for the numerical examples.

\bibliographystyle{erae}
\bibliography{ppi}

\end{document}